\documentclass[superscriptaddress,aps,twocolumn, showpacs,nofootinbib,longbibliography,notitlepage]{revtex4-2}
\usepackage{etex}
\usepackage{physics}
\usepackage{amsmath,amssymb,amsthm,amstext,mathrsfs}
\DeclareMathAlphabet{\mathscrbf}{OMS}{mdugm}{b}{n}
\DeclareMathAlphabet\mathbfcal{OMS}{cmsy}{b}{n}
\usepackage[colorlinks=true,citecolor=blue,urlcolor=blue]{hyperref}
\usepackage[pdftex]{graphicx}
\usepackage{times,txfonts}
\usepackage{braket}
\usepackage{color}
\usepackage{natbib}
\usepackage{amsmath,blkarray}
\usepackage{mathtools,mathrsfs}
\usepackage{latexsym}
\usepackage{tabularx, booktabs}
\usepackage{graphics,epstopdf}
\usepackage{graphicx}
\usepackage{float}
\usepackage{graphicx}
\usepackage{amsfonts}
\usepackage{subcaption}
\usepackage{enumerate}
\usepackage{color,soul}
\usepackage{bbold}
\usepackage{centernot}

\newtheorem{cor}{Corollary}

\newtheorem{defn}{Definition}
\newtheorem{proposition}{Proposition}

\newtheorem{thm}{Theorem}

\newcommand*{\Comb}[2]{{}^{#1}C_{#2}}

\DeclareMathOperator*{\Motimes}{\text{\raisebox{0.25ex}{\scalebox{0.65}{$\bigotimes$}}}}

\begin{document}

\title{Towards Necessary and sufficient state condition for violation of a multi-settings Bell inequality}

\author{Swapnil Bhowmick}
\email{swapnilbhowmick@hri.res.in}
\affiliation{ Quantum Information and Computation Group,
Harish-Chandra Research Institute, A CI of Homi Bhabha National Institute, Chhatnag Road, Jhunsi, Prayagraj 211019, India}

\author{Som Kanjilal}
\email{som.kanjilal@inl.int}
\affiliation{International Iberian Nanotechnology Laboratory (INL), Av. Mestre José Veiga, 4715-330 Braga, Portugal}

\author{A. K. Pan }
\email{akp@phy.iith.ac.in}
\affiliation{Indian Institute of Technology Hyderabad, Kandi, Sangareddy, Telengana 502285, India}

\author{Souradeep Sasmal}
\email{souradeep.007@gmail.com}
\affiliation{Institute of Fundamental and Frontier Sciences, University of Electronic Science and Technology of China, Chengdu 611731, China}

\begin{abstract}
High-dimensional quantum entanglement and the advancements in their experimental realization provide a playground for fundamental research and eventually lead to quantum technological developments.  The Horodecki criterion determines whether a state violates Clauser-Horne-Shimony-Holt (CHSH) inequality for a two-qubit entangled state, solely from the state parameters. However, it remains a challenging task to formulate similar necessary and sufficient criteria for a high-dimensional entangled state for the violation of a suitable Bell inequality. Here, we develop a Horodecki-like criterion based on the state parameters of arbitrary two-qudit states to violate a two-outcome Bell inequality involving $2^{n-1}$ and $n$ measurement settings for Alice and Bob, respectively. This inequality reduces to the well-known CHSH and Gisin's elegant Bell inequalities for $n=2$ and $n=3$, respectively. While the proposed criterion is sufficient to violate the Bell inequality,  it becomes necessary as well for the following cases; (i) $m$ copies of Bell diagonal states for arbitrary $n$, (ii) Non-decomposable states whose correlation matrix is diagonalized by local unitaries, and (iii) for any arbitrary two-qubit state when $n=3$, where the maximal value of the Bell functional is achieved with Bob’s measurements being pairwise anticommuting. For any states, we derive the constraints on Alice's measurements in achieving the maximum quantum violation for this inequality. 
\end{abstract}

\maketitle

\section{Introduction} \label{intro}
Bell theorem asserts that certain quantum statistics generated by two distant parties (Alice and Bob) cannot be mimicked by a classical model satisfying the notion of local realism \cite{Bell1964}. This characteristic, widely referred to as Bell nonlocality \cite{Brunner2014}, is commonly demonstrated through the quantum violation of suitable Bell inequalities. In addition to challenging our conventional understanding about nature \cite{Spekkens2005, Pusey2012, Popescu2014, Bong2020, Renou2021}, Bell nonlocality has found a plethora of applications \cite{Brunner2014} in information processing tasks, ranging from cryptography \cite{Ekart1991, Acin2007}, communication complexity \cite{Brukner2004}, certified randomness generation \cite{Colbeck2012, Pironio2010, Wooltorton2022} to game theory \cite{Brunner2013a} and computation \cite{Brunner2014, Supic2020rev}. 

Note that entanglement is a necessary condition for demonstrating the Bell nonlocality, but it is not sufficient. Although Bell non-locality can be demonstrated for all pure entangled states \cite{Yu2012}, the case of mixed entangled states is not so straightforward. There exist mixed entangled states admitting local ontic models \cite{Werner1989, Augusiak2014, Bowles2015a,Hirsch2016}, therefore not violating any Bell inequality. This necessitates a criterion based on state parameters to characterize the class of nonlocal mixed entangled states.

One may surmise that a natural approach to framing such criteria involves determining the interrelationship between local observables for which a Bell functional attains the optimal quantum value. By doing so, we can intuitively fix the nature of these observables. For instance, the optimal quantum value of the CHSH function \cite{Clauser1969} is $2\sqrt{2}$, achieved with local anticommuting observables. One might expect that for any two-qubit state, the corresponding maximum CHSH value can be determined by optimizing over anticommuting local observables. However, Horodecki et al. \cite{Horodecki1995} demonstrated that this approach does not always work, as there are two-qubit states for which the maximum CHSH value is achieved when one party's observables are not anticommuting. 

Hence, deriving the violation of a given Bell inequality solely from the state parameters is a challenging proposition. This underscores the novelty of the Horodecki criterion \cite{Horodecki1995} that provides a \textit{necessary and sufficient} criterion in terms of the state parameters to determine whether a given two-qubit state will violate the CHSH inequality. Otherwise, one would need to perform the optimisation over all possible qubit observables to obtain the CHSH violation for a given state. The Horodecki criterion also determines the exact choice of observables in terms of state parameters. This criterion has been extended to non-projective measurements for bipartite qubit states \cite{Hall2022} to violate CHSH inequality and for the tripartite qubit case \cite{Siddiqui2022} to violate the Svetlichney inequality \cite{Svetlichny1987}. 

Recently, it has been brought to attention that higher dimensional quantum states offer more resources in a wide array of information processing tasks, including key distribution \cite{Cerf2002, Georgios2005, Mafu2013, Sekga2023}, random number generation \cite{Ma2016, Nie2014, Bai2021, Mannalatha2023}, and fault-tolerant quantum computation \cite{Gottesman1999, Omanakuttan2023, Kapit2016, omanakuttan2024}. Additionally, with an increasing number of experimental realizations of higher-dimensional quantum states \cite{Kues2017, Hu2020, He2022, Min-Hu2020}, they are becoming more sought after than ever before. These advancements immediately call for characterizing higher-dimensional mixed entangled states demonstrating the Bell nonlocality. As the local dimensions increase, the complexity of the problem grows exponentially. For instance, if the local dimension consists of an $m$-qubit system, each observable contains $\sim 2^{m}$ parameters, and the state space of an $m$-qubit system forms a $(2^{2m}-1)$-dimensional sphere \cite{NielsenChuang2010}.

In this work, by going beyond bipartite two-qubit states, we consider higher-dimensional bipartite states of local dimension $d=2^{m}$. Note that in higher dimensions, there are two broader classes of entangled states: The two-qudit state is either decomposable into $m$ copies of two-qubit states or not \cite{Kraft2018}. It is important to emphasize that, not being a dimension witness\footnote{A Bell inequality serves as a dimension witness if both its local and optimal quantum bounds are functions of the dimension of the shared state \cite{Brunner2008, Dallarno2012, Brunner2013, Guhne2014}.}, the CHSH inequality is inadequate to reveal the nonlocality of bipartite higher-dimensional states \cite{Brunner2008, Dallarno2012, Brunner2013, Guhne2014}. Therefore, we need to employ either a higher setting two-outcome Bell inequality that serves as a dimension witness \cite{Ghorai2018} or introduce higher-outcome Bell inequalities \cite{Salavrakos2017}. 

Here, we invoke a particular class of multi-settings, two-party, two-outcome Bell inequality \cite{Ghorai2018}, which has been shown as dimension witness \cite{Pan2020}. This inequality involves the measurement settings $2^{n-1}$ and $n$ for Alice and Bob, respectively. This inequality reduces to the well-known CHSH inequality \cite{Clauser1969} and Gisin's Elegant Bell inequality \cite{Gisin2007} for $n=2$ and $n=3$, respectively. The optimal quantum value of this inequality is achieved with at least $m$ copies of a maximally entangled two-qubit state \cite{Mahato2023} (or, equivalently \cite{Kraft2018}, a maximally entangled two-qudit state of local dimension $d=2^m$, where $m=\lceil \frac{1}{2}(n-1)\rceil$). Additionally, this inequality has found applications in a parity oblivious random-access-code \cite{Ghorai2018} and device-independent randomness generation scheme \cite{Mahato2023}.

In particular, we aim to develop a Horodecki-like criterion for a given two-qudit state to violate the $n$-settings Bell inequality.  For this purpose, it is necessary to optimize the $\sim 2^{n-1}(2^{n-1}+n)$ variables to determine the optimal value of the $n$-settings Bell inequality, which scales exponentially with $n$. To address this challenge, we begin by exploiting the structure of the inequality to prove that specific linear combinations of Alice's observables must be anticommuting to achieve the optimal quantum value. This finding allows us to restrict the domain of observables to those that are anticommuting on Alice's side, thereby reducing the complexity of the optimization problem to a linear function of $n$. Subsequently, by suitably defining a Horodecki-like function $M_n(\rho)$ in terms of state parameters, we show that the optimal Bell value for a given state is tightly lower bounded by $M_n(\rho)$, constituting a sufficient criterion for the violation of $n$-settings Bell inequality. This means that if $M_n(\rho)$ exceeds the local bound, the state is nonlocal; otherwise, the nature of the state remains inconclusive. 

Crucially, we further demonstrate that $M_n(\rho)$ saturates the optimal Bell value and thus becomes the necessary and sufficient criterion for the following two classes of states (i) $m$ copies of Bell diagonal states for arbitrary $n$, (ii) non-decomposable states whose correlation matrix can be diagonalized by local unitaries, and (iii) for any arbitrary two-qubit state when $n=3$. Interestingly, for these cases, we find that corresponding maximal values of the Bell inequality occur if Bob's measurements are mutually anticommuting.

The paper is structured as follows. To begin with, we briefly discuss the $n$-settings Bell inequality \cite{Ghorai2018} and the structure of an arbitrary two-qudit state in Sec.~\ref{pre}. We then state the precise constrained optimization problem corresponding to finding the Horodecki-like criteria for the $n$-settings inequality in Sec.~\ref{sop}. Our main result, given by Theorem \ref{theo}, is a sufficiency condition for Bell violation, which involves optimization over an exponentially reduced number of variables. The proof utilizes the fact that specific linear combinations of Alice's observables must be anticommuting for optimal violation, as stated in proposition \ref{prop}. Interestingly, we show that there exists a certain class of states for which our sufficient condition also becomes necessary as well in Corollaries \ref{cor1}, \ref{cor2} and \ref{cor3}. We also construct Alice and Bob's observables that achieve the sufficiency condition given in Theorem \ref{theo} for a given state (Sect.~\ref{conobs}). Subsequently, we relate our result with Horodecki's original work, demonstrating how the latter can be obtained as a special case. To conclude, we summarise our findings and provide a brief discussion about future work in Sec.~\ref{outlook}.


\section{Preliminaries}\label{pre}
We begin with some useful definitions and terminology starting with the $n$ settings of the Bell inequality first introduced in \cite{Ghorai2018}. Then, we discuss the Fano form of representing bipartite states with arbitrary local dimensions before defining the problem formally.

\subsection{A family of multi-settings Bell inequalities} In a bipartite Bell scenario, Alice performs the measurement of a set of $2^{n-1}$ dichotomic observables  $A:=\{A_{x}: x\in \{1,2,\ldots,2^{n-1}\}\}$, and Bob performs a set of $n$  dichotomic observables  $B:=\{B_y: y\in \{1,2,\ldots,n\}\}$. The measurements of both parties produce outcomes $\pm 1$. In this scenario, we consider a Bell functional as follows.
\begin{equation} \label{nbell}
\mathcal{I}_n=  \sum\limits_{y=1}^{n} \qty(\sum\limits_{x=1}^{2^{n-1}} (-1)^{z_y^{x}} A_x)  \otimes B_y 
\end{equation} 
where ${z_{y}^{x}}$ is the $(n+1-y)^{th}$ bit of the number $(x-1)$ expressed in binary. Note that for $n=2$ and $n=3$, Eq.~(\ref{nbell}) reduces to the well known CHSH inequality \cite{Clauser1969} and Gisin's elegant Bell inequality \cite{Gisin2007} respectively. The local bound of $\mathcal{I}_n$ is derived \cite{Munshi2021} as
\begin{eqnarray}
(\mathcal{I}_n)_{L}=\qty(\Big\lfloor \frac{n}{2} \Big\rfloor +1) \binom{n}{\lfloor \frac{n}{2}\rfloor +1} 
\end{eqnarray}  
with $\lfloor\cdot\rfloor$ is the floor function.

It has been demonstrated \cite{Ghorai2018} that the optimal quantum value of the Bell functional, $\mathcal{I}_n$, given by Eq.~(\ref{nbell}), is $(\mathcal{I}_n)_{\mathcal{Q}}^{opt}=2^{n-1}\sqrt{n}$. This value is attained for maximally entangled states under the condition that at least $n$ mutually anticommuting observables exist within the local Hilbert space of both parties. Hence, for a given $n$, the minimal dimension of the local Hilbert space has to be $d=2^{m}$ with $m=\lceil \frac{1}{2}(n-1) \rceil$ for maximum violation. Consequently, this feature is used in \cite{Pan2020} to certify at least $m=\lceil \frac{1}{2}(n-1) \rceil$ numbers of maximally entangled two-qubit states. 

For our purpose, we re-write the quantum value of the Bell functional in Eq.~(\ref{nbell}) in terms of the scaled observables for Alice $\mathcal{A}:=\qty{\mathcal{A}_{y}: y\in [1,2,\ldots,n]}$, defined as follows.
\begin{equation}\label{obs3}
  A'_y = \sum\limits_{x=1}^{2^{n-1}}(-1)^{z^x_y} \ A_x ; \ \ \mathcal{A}_y = \frac{A'_y}{\omega_y} ;  \ \ \omega_{y} = \sqrt{\Tr[(A_y^{'})^2\rho]}
\end{equation}
Note that the coefficients $(-1)^{z^x_y}$ appearing above can be represented by the matrix element ($\mathscr{M}_{yx}$) corresponding to $y^{th}$ row and $x^{th}$ column of a orthogonal matrix $\mathscr{M}$ defined in Appx. \ref{nproof}. Putting Eq.~(\ref{obs3}) in Eq.~(\ref{nbell}) gives
\begin{eqnarray}\label{bellvalue}
\mathcal{I}_n (\mathcal{A},B,\rho)=  \sum_{y=1}^{n}\omega_{y}\Tr\left[ (\mathcal{A}_{y}\otimes B_{y})\rho\right] 
\end{eqnarray}
We now describe the Fano form for representing quantum states which we use in the rest of the paper.


\subsection{Representation of $m$-copies of two-qubit state} \label{rots}

Any arbitrary two-qubit state $\psi \in \mathscr{L}(\mathscr{H}_A^2 \otimes \mathscr{H}_B^2)$, given in any arbitrary operator basis $\{\upsilon_i\otimes\upsilon_j'\}$, can always be diagonalised with local unitary $U \otimes V$ into the diagonal operator basis $\sigma=U \upsilon U^{\dagger}$ and $\sigma'=V\upsilon' V^{\dagger}$, and expressed as follows \cite{Horodecki1995}
\begin{equation}\label{Fanoqubit}
\psi =\frac{1}{4}\qty[\openone + \sum_{i=1}^{3} r_i \sigma_i \otimes \openone + \sum_{j=1}^{3} s_j \openone\otimes\sigma'_j +\sum_{k=1}^{3} \lambda_k \sigma_k \otimes\sigma_k']
\end{equation}
$\sigma$ and $\sigma'$ are the generators of $SU(2)$ on Alice's and Bob's subspace, respectively. Now $m$ copies of $\psi$, given by $\rho=\Motimes_m \psi \in \mathscr{L}(\mathscr{H}_A^d \otimes \mathscr{H}_B^d)$, with $d=2^m$, in general, can always be expressed in a particular basis $\{\tau,\tau'\}$ as \cite{Fano1957}
\begin{equation}\label{Fanoqudit}
\rho=\frac{1}{d^2}\qty[\openone_{d^2}+ \sum_{i=1}^{d^2 -1} p_i \tau_i \otimes \openone + \sum_{j=1}^{d^2 -1} q_j \openone\otimes\tau'_j +\sum_{u,v=1}^{d^2 -1} t_{uv}\tau_u \otimes\tau'_v] 
\end{equation}
where the chosen operator basis set $\boldsymbol{\tau}:=\{\tau_k\}_{k=1}^{d^2 -1}$ given as follows
\begin{equation} \label{taubasis}
    \tau_u = \sigma_{a_1} \otimes \sigma_{a_2} \otimes \cdots \otimes\sigma_{a_m}  \ ; \ \ \tau'_v = \sigma'_{b_1} \otimes \sigma'_{b_2} \otimes \cdots \otimes \sigma'_{b_m} \ ;
\end{equation}
where $a_k, b_l \in \{0,1,2,3\}$ with $\sigma_0 = \openone$ and $\{\sigma_1, \sigma_2, \sigma_3 \}$ are the generators of $SU(2)$. There are $(4^m -1)$ number of $m$ quart-string $(a_1, a_2,\cdots, a_m)$ if we exclude $a_i =0 \ \forall i$. For each $u \in \{1,2,\cdots d^2 -1 \}$ with $d=2^m$, $\tau_u$ corresponds to one $m$ quart-string. A similar definition also holds for $\tau'_v$. Thus, $\boldsymbol{\tau}$ are a set of traceless, dichotomic and Hermitian operators. The correlation matrix, $T_{\rho}:= [t_{uv}]$ is defined as the $(d^2 -1)\times(d^2 -1)$ matrix whose elements are given by $t_{uv}= \Tr [\tau_u \otimes \tau'_v \ \rho]$. The joint operators appearing in diagonal positions, i.e., $\{\tau_u \otimes \tau'_u \}$ form a mutually commuting set. Note that the basis set $\boldsymbol{\tau}$ contains at most $(2m+1)$ mutually anticommuting operators (see Appx.~G of \cite{Monroig2020}). In addition, there exist distinct subset $\mathcal{S}_l \subset \boldsymbol{\tau}$ of $(2m+1)$ mutually anticommuting operators. For example, when $\tau \in \mathscr{L}(\mathscr{H}^4)$, there are 6 distinct sets of 5 exhaustive mutually anticommuting observables.


\section{Statement of the problem} \label{sop}

Given $m$ copies of the two-qubit state $\psi$, we intend to identify a precise condition in terms of state parameters that determine whether the given state $\rho=\otimes_m \psi$ violates the $n$-settings Bell inequality in Eq.~(\ref{nbell}).

One strategy for formulating such a condition involves evaluating the maximum Bell value, denoted by $\mathcal{I}_n(\rho)$ corresponding to the state as a function of its parameters, by optimizing over all relevant observables. If this optimal value exceeds $(\mathcal{I}_n)_{L}$, it indicates nonlocality.

To achieve this, we solve the following optimization problem
\begin{equation}\label{obj}
\begin{aligned}
\mathcal{I}_n (\rho)&= \max_{\{\mathcal{A}, B\}} \sum_{y=1}^{n}\omega_{y}\Tr[(\mathcal{A}_{y}\otimes B_{y})\rho] \\
\textrm{subject to:} \ \ B_{y}^{2}&=\openone \ \ \ \forall y\in \{1,2,\ldots,n\}   \\
A_x^2 &= \openone \ \ \ \forall x\in \{1,2,\ldots, 2^{n-1}\}
\end{aligned}
\end{equation}
Here, the maximisation is performed over sets $A:=\{A_x\}$ and $B:=\{B_y\}$, comprising dichotomic Hermitian operators with eigenvalues $\{\pm1\}$. We denote the set of joint observables over which the optimization is performed as $\Omega$. To present our results, we introduce the function $M_n(\rho)$. 
\begin{defn}\label{basis}
Given a state $\rho=\otimes_m \psi$ in the operator basis $\boldsymbol{\tau}\otimes \boldsymbol{\tau}'$ of the from Eq.~(\ref{taubasis}), we define 
    \begin{equation}\label{mtau}
    M_n \qty(\rho)= \max_{\{\mathcal{S}^{(n)}_{l} , \mathcal{S}_{l'}\}} \sqrt{\sum_{\tau_u \in \mathcal{S}^{(n)}_{l}}\sum_{\tau'_v \in\mathcal{S}_{l'}} \braket{\tau_u\otimes \tau'_v}_{\rho}^2}
\end{equation}
\end{defn}
where $\mathcal{S}^{(n)}_{l}\subseteq \mathcal{S}_{l}$ contains $n$ mutually anticommuting operators with $n\leq (2m+1)$. The sets $S^{(n)}_l$ and $S_{l'}$ contain $n$ and $(2m+1)$ elements respectively. The maximum in Eq.~(\ref{mtau}) is taken over all possible choices of sets $\{\mathcal{S}^{(n)}_l\}_l$ and $\{\mathcal{S}_{l'}\}_{l'}$ from $\boldsymbol{\tau}$ and $\boldsymbol{\tau'}$.

Note that $m$ copies of $\psi$ may be provided in an arbitrary basis $\otimes_m \{\upsilon_i \Motimes \upsilon'_j\}$. Since $\psi$ can be diagonalized using $U \otimes U'$, the form of $\otimes_m \psi$ can always be achieved as in Eq.~(\ref{Fanoqudit}) by applying local unitary transformations by $\otimes_m U$ and $\otimes_m U'$. This ensures the given $m$ copies of the two-qubit state can be expressed in the operator basis $\boldsymbol{\tau}\otimes \boldsymbol{\tau}'$ as defined in Eq.~(\ref{taubasis}).

Our aim is to establish a relationship between the function $M_n (\rho)$ and $\mathcal{I}_n (\rho)$, enabling comparison of $M_n (\rho)$ with $(\mathcal{I}_n)_{L}$ to ascertain the nonlocality of the state.


\section{Results} \label{results}

 In this section, in Theorem \ref{theo}, we derive a sufficient condition, expressed in terms of $M_n (\rho)$, for the violation of the $n$-settings Bell inequality. We also prove the necessity of the existence of $n$ maximally incompatible measurements on Alice's Hilbert space for attaining the maximal Bell value for any given state, as shown in Proposition \ref{prop}. Subsequently, we explore the cases where $M_n (\rho)$ becomes necessary as well, as detailed in Corollaries \ref{cor1}, \ref{cor2} and \ref{cor3}. Furthermore, we provide a construction of observables that achieve the value of $M_n (\rho)$ for $n=3$ (Sec.~\ref{conobs}). Finally, we demonstrate how our construction of $M_n (\rho)$ reduces to Horodecki's definition of $M(\rho)$ for the CHSH inequality ($n=2$).

\begin{thm} \label{theo}
For a given state $\rho$, a sufficient condition for the violation of the $n$-settings Bell inequality in Eq.~(\ref{nbell}) is given by
\begin{equation}\label{ub2}
M_n(\rho) > 2^{1-n} \ (\mathcal{I}_n)_{{L}}
\end{equation}
\end{thm}

Eq.~(\ref{ub2}) underscores that the correlation matrix is sufficient to determine the nonlocality of a state. We later show that in some special cases, the eigenvalues of the correlation matrix provide the necessary and sufficient conditions to determine nonlocality.
 
A more practical implication of Eq.~(\ref{ub2}) is its potential to exponentially reduce the complexity of the optimization problem defined in Eq.~(\ref{obj}). It is straightforward to check that each of $\{A_x \}$ and $\{B_y\}$ contains $d^2 -d$ numbers of independent parameters with $d=2^m$ and $m=\lceil \frac{1}{2}(n-1) \rceil$. There are $2^{n-1} +n$ observables. Hence, the optimization problem involves $(d^2-d)(2^{n-1} +n) \sim \mathscr{O}(2^n)$ parameters for large $n$. On the other hand, the optimization in the definition of $M_n(\rho)$ given by Eq.~(\ref{mtau}) involves only a finite number of sets of $n$ mutually anticommuting operators, i.e., $\mathcal{S}_l^{(n)}$ and $\mathcal{S}_{l'}$. Once the state is given and expressed in the desired basis $\{\tau\otimes \tau'\}$, all sets $\mathcal{S}_l^{(n)}$ and $\mathcal{S}_{l'}$ are fixed.

\begin{proof}
Applying the Cauchy-Schwarz inequality to Eq.~(\ref{bellvalue}), we obtain
\begin{eqnarray} \label{nsettingcauchy}        
\braket{\mathcal{I}_n (\mathcal{A},B,\rho)} &=&  \sum_{y=1}^{n} \omega_y \ \Tr[\rho \ \qty(\mathcal{A}_y \otimes B_y )] \nonumber\\
        &\leq& \sqrt{\sum_{y=1}^n \omega_y^2}\sqrt{\sum_{y=1}^n \braket{\mathcal{A}_y \otimes B_y}^2} 
\end{eqnarray}
Next, we evaluate the following upper bound for the quantity $\sqrt{\sum_{y=1}^n \omega_y^2}$, which is saturated for an arbitrary state $\rho$.
\begin{proposition} \label{prop}
$\sqrt{\sum_{y=1}^n \omega_y^2} \leq 2^{n-1}$, which is saturated when Alice's scaled observables satisfy $\{\mathcal{A}_y,\mathcal{A}_{y'}\}=0 \ \forall y \neq y' \in \{1,2,...,n\}$, independent of the given state $\rho$.
\end{proposition}
The detailed proof is extensive and, therefore, deferred to Appx. \ref{parity obliv}. Here we explain the logical reasoning to construct the proof. We make use of the fact that within the definition of $A':=\{A'_y\}$ in Eq.~(\ref{obs3}), coefficients of $A_x$ can be arranged into rows of an orthogonal matrix $\mathscr{M}$ (see Appx. \ref{nproof}). Concretely, we show that the transformation from $A$ to $A'$ is governed by an orthogonal matrix $\mathscr{M}$ constructed from the scalar elements $(-1)^{z^{x}_{y}}$ defined in Eq.~(\ref{nbell}). Using the orthogonality of $\mathscr{M}$ and $A_{x}^{2}=\openone$, we demonstrate that the maximum value of $\sqrt{\sum_{y=1}^n \omega_y^2}$ is attained when a specific set of state-independent constraints (given by Eq.~\ref{alicecondn} in Appx. \ref{nproof}) are adhered to by Alice's set of observables $A$. These constraints consequently lead to the anticommutation relations $\{\mathcal{A}_y,\mathcal{A}_{y'}\}=0$ which then implies $\max_{\mathcal{A}} \sqrt{\sum_{y=1}^n \omega_y^2}=2^{n-1}$ as shown in Appx. \ref{nproof}.

One consequence of the anticommutation relations $\{\mathcal{A}_y,\mathcal{A}_{y'}\}=0$ is that $\mathcal{A}$ must be traceless, as can be shown. Multiplying the anticommutation relation on both sides by $\mathcal{A}_y^{-1}$, we obtain $\mathcal{A}_y^{-1} \{\mathcal{A}_y,\mathcal{A}_{y'}\}= \mathcal{A}_{y'} + \mathcal{A}_y^{-1} \mathcal{A}_{y'} \mathcal{A}_y$. Taking trace on both sides, we get $\Tr[\mathcal{A}_y]=0 \ \forall y$, entailing $\Tr[A_x]=0 \ \forall x$. 

The necessity of characterizing $\mathcal{A}$ as the set of $n$ pairwise anticommuting traceless operators enables us to map them to one of the anticommuting subsets $\mathcal{S}^{(n)}_l$ of the basis $\boldsymbol{\tau}$ introduced in Eq.~(\ref{taubasis}) with $\tau_u\equiv\mathcal{A}_u \ \forall \tau_u \in \mathcal{S}^{(n)}_l$ for some $l$. A significant consequence is that, since Alice's observables must be linear combinations of anticommuting observables to attain the maximal value, this finding reduces the joint observable search space parameters to a strict subset $\mathcal{K} \subset \Omega$. Consequently, the optimization problem of Eq.~(\ref{obj}) reduces to the following optimization problem.
\begin{equation} \label{objred}
\max_{\{\mathcal{A}, B\}} \sum_{y=1}^{n}\omega_{y}\Tr\left[ (\mathcal{A}_{y}\otimes B_{y})\rho\right]  \to 2^{n-1}\max_{\mathcal{S}^{(n)}_{l} ,\{B\}}   \sqrt{\sum_{\tau_y\in \mathcal{S}^{(n)}_{l}} \braket{\tau_{y}\otimes B_{y}}_{\rho}^{2}}   
\end{equation}
The maximization over the subsets $\mathcal{S}^{(n)}_{l}$ remains because there are a finite number of subsets and we have to take into account every possible choice of $\mathcal{S}^{(n)}_{l}$ that is fixed for a given state. 

Now we consider $B_{y}$ as a linear combination of anticommuting basis elements $\tau'_v\in\boldsymbol{\tau}'$ as $B_{y}=\sum_{v\in\mathcal{S}_{l}} b_{yv}\tau'_{v}$ with $\sum_{v}b_{yv}^{2}=1$. This consideration ensures $B_{y}^{2}=\openone$. Subsequently, following a similar argument used for $\mathcal{A}$, $B_y$ becomes traceless. It is worth noting that the complete space of dichotomic observables may contain operators that cannot be represented as a linear combination of pairwise anticommuting operators. Since this way of expressing $B_y$ only allows optimization over a subset of dichotomic observables on Bob's side, this assumption further reduces the joint observable search space to $\mathcal{R}\subset \mathcal{K}$ and yields a tight lower bound of the optimization of RHS of Eq.~(\ref{objred}).
\begin{equation} \label{objred3}
\max_{\mathcal{S}^{(n)}_{l},\{B\}}   \sqrt{\sum_{\tau_y\in \mathcal{S}^{(n)}_{l}} \braket{\tau_{y}\otimes B_{y}}_{\rho}^{2}} 
 \geq \max_{\mathcal{S}^{(n)}_{l},\mathcal{S}_{l'}} \sqrt{\sum_{\tau_y\in \mathcal{S}^{(n)}_{l}} \qty( \sum_{\tau'_v\in\mathcal{S}_{l'}}b_{yv} \ \braket{\tau_{y}\otimes\tau'_{v}}_{\rho})^{2}}
\end{equation}
The bound is tight because it can be saturated for some classes of states which we will see later. To evaluate the upper bound of RHS of Eq.~(\ref{objred3}), we apply the Cauchy-Schwartz inequality and using $\sum_{\nu}b_{y\nu}^{2}=1$, we obtain the upper bound of the quantity $\sum_{\tau_y\in \mathcal{S}^{(n)}_{l}} \qty( \sum_{\tau'_v\in\mathcal{S}_{l'}}b_{yv} \ \braket{\tau_{y}\otimes\tau'_{v}}_{\rho})^{2}$ as follows
\begin{equation}\label{cs11}
 \sqrt{\sum_{\tau_y\in \mathcal{S}^{(n)}_{l}} \qty( \sum_{\tau'_v\in\mathcal{S}_{l'}}b_{yv} \braket{\tau_{y}\otimes\tau'_{v}}_{\rho})^{2}} \leq \sqrt{\sum_{\substack{\tau'_v\in\mathcal{S}_{l'} \\ \tau_y\in \mathcal{S}^{(n)}_{l}}} \braket{\tau_{y}\otimes\tau'_{v}}_{\rho}^2 } 
\end{equation}
Above inequality in Eq.~(\ref{cs11}) saturates for the following condition on the parameter $b_{yv}$
 \begin{eqnarray}\label{obsn}
b_{yv}=\frac{\braket{\tau_{y}\otimes\tau'_{v}}_{\rho}}{\sqrt{\sum_{\tau'_v\in\mathcal{S}_{l'}}\braket{\tau_{y}\otimes\tau'_{v}}_{\rho}^{2}}}
 \end{eqnarray}

Using Eq.~(\ref{cs11}), we obtain an upper bound of RHS of Eq.~(\ref{objred3}) as
\begin{equation}\label{newbound}
    \max_{\mathcal{S}^{(n)}_{l}  \mathcal{S}_{l'} }\sqrt{\sum_{\tau_y\in \mathcal{S}^{(n)}_{l}} \qty( \sum_{\tau'_v\in\mathcal{S}_{l'}}b_{yv} \ \braket{\tau_{y}\otimes\tau'_{v}}_{\rho})^{2}} \leq \max_{\mathcal{S}^{(n)}_{l}\mathcal{S}_{l'}}\sqrt{\sum_{\substack{\tau'_v\in\mathcal{S}_{l'} \\ \tau_y\in \mathcal{S}^{(n)}_{l}}} \braket{\tau_{y}\otimes\tau'_{v}}_{\rho}^2 }
\end{equation}

\begin{figure}[!ht] 
\includegraphics[scale=0.47]{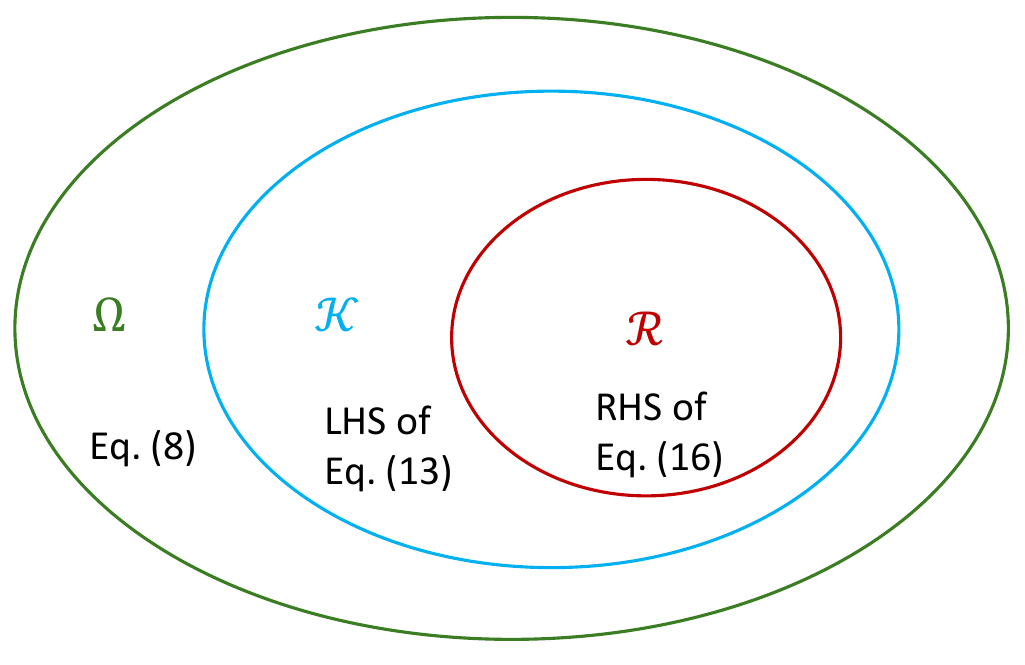}
\caption{$\Omega\equiv$ The set of joint observables $(\boldsymbol{\mathcal{A}}\otimes B)$ that evaluates $\mathcal{I}_n(\rho)$ in Eq.~(\ref{obj}). $\mathcal{K}\equiv$ The set of joint observables in which Alice's observables are linear combinations of mutually anticommuting operators. The solution of LHS of Eq.~(\ref{objred3}) comes from this set of observables. $\mathcal{R}\equiv$ The set of joint-observables where both Alice's and Bob's observables are expressed in linear combinations of anticommuting operators. This set of observables provides the solution RHS of Eq.~(\ref{newbound}). Note that the value of $\mathcal{I}_n(\rho)$, optimized over the set $\mathcal{R}$ will always be less than or equal to the value evaluated optimized over the set $\mathcal{K}$. Proposition \ref{prop} proves that the solution of the optimization problem in Eq.~(\ref{obj}) comes from $\mathcal{K}$.}
\label{obsdiaset}
\end{figure}

Here, we emphasize that the upper bound of RHS of Eq.~(\ref{objred3}) is achieved by invoking the Cauchy-Schwartz inequality, which implies that the optimal value of RHS of Eq.~(\ref{newbound}) is attained by observables in $\mathcal{R}\subset \mathcal{K}$. Since Proposition \ref{prop} guarantees that the value of $\mathcal{I}_n (\rho)$ (LHS of Eq.~\ref{objred3}) must be obtained from the set $\mathcal{K}$, the maximum value of RHS of Eq.~(\ref{newbound}) is lower than that of LHS of Eq.~(\ref{objred3}). This characteristic is captured by the following relations (see Fig. \ref{obsdiaset}).

\begin{eqnarray}
 \mathcal{I}_n(\rho) &\equiv& 2^{n-1}\max_{\mathcal{S}^{(n)}_{l} ,\{B\}}   \sqrt{\sum_{\tau_y\in \mathcal{S}^{(n)}_{l}} \braket{\tau_{y}\otimes B_{y}}_{\rho}^{2}} \nonumber \\
 &\geq & \ 2^{n-1} \max_{\mathcal{S}^{(n)}_{l},\mathcal{S}_{l'}}  \sqrt{\sum_{\tau_y\in \mathcal{S}^{(n)}_{l},\tau'_v\in\mathcal{S}_{l'}}\braket{\tau_{y}\otimes\tau'_{v}}_{\rho}^2 } \nonumber \\
 &=& \ 2^{n-1}  M_n(\rho)  \label{cs2}
\end{eqnarray}
\end{proof}

From Eq.~(\ref{cs2}), we observe that $ M_n (\rho)$ is upper bounded by $2^{1-n} \mathcal{I}_n(\rho)$. This implies that if $ M_n (\rho)>2^{1-n}(\mathcal{I}_n)_{L}$ the given state $\rho$ is nonlocal, thereby establishing $M_n(\rho)$ as a sufficient criterion. Now, let us analyse under what condition $M_n (\rho)$ becomes both a necessary and sufficient condition, that is, when $M_n(\rho)=2^{1-n}\mathcal{I}_n(\rho)$.  The following Corollary brings out the reasoning for the consideration of the basis $\boldsymbol{\tau}\otimes\boldsymbol{\tau}'$, which saturates Eq.~(\ref{cs2}).

\begin{cor}\label{cor1}
For a decomposable Bell diagonal state\footnote{$\rho=\psi^{\otimes m}_D = \frac{1}{2^{2m}} \qty[\openone\otimes\openone + \sum_{i=1}^3 \lambda_i \ \qty(\sigma_i \otimes \sigma'_i)]^{\Motimes_m}$}, denoted by $\rho_{BD}=\otimes_m\psi_{BD}$, the inequality in Eq.~(\ref{cs2}) saturates, i.e., $M_n(\rho_{BD})=2^{1-n}\mathcal{I}_n(\rho_{BD})$. In this case, the necessary and sufficient condition for violating the $n$-settings Bell inequality in Eq.~(\ref{nbell}) is given by
\begin{equation}\label{belldiag}
M_n(\rho_{BD})=\max_{\mathcal{S}^{(n)}_l} \sqrt{\sum_{i=1}^n \mu_i} > 2^{1-n}(\mathcal{I}_n)_{L}
\end{equation}
with Bob's observables are mutually anticommuting, and the basis set $\boldsymbol{\tau}\otimes\boldsymbol{\tau}'$ becomes the diagonal operator basis of $\rho_{BD}$. $\{\mu_i\}$ are those eigenvalues out of $(d^2-1)$ eigenvalues of $T^{\top}_{\rho_{BD}} \cdot T_{\rho_{BD}}$, such that $\mu_i =\braket{\tau_i \otimes\tau'_i}^2_{\rho_{BD}}$ with $\tau_i, \tau'_i \in \mathcal{S}_l^{(n)}$.   
\end{cor}

\begin{proof}
For Bell diagonal states, the non-zero contributions to $M_n(\rho)$ come from the diagonal elements. This makes sure that the sets $\mathcal{S}^{(n)}_l$ and $S_{l'}$ are formed with operators corresponding to diagonal positions, implying $\mathcal{S}^{(n)}_l=S_{l'}$. Using Proposition~\ref{prop}, if we map $\mathcal{A}$ in $\mathcal{S}^{(n)}_l$, it automatically results in the mapping $B\in S_{l'}$. The eigenvalue interlacing Theorem\footnote{Eigenvalue interlacing Theorem \cite{Hwang2004} states that eigenvalues ($\{\gamma_i : \gamma_1 \leq \gamma_2 ...\leq \gamma_n\}$) of a $n\times n$ sub-block, within a larger $d\times d$ symmetric matrix $M$ with eigenvalues $\{\mu_i : \mu_1 \leq \mu_2 ...\leq \mu_d\}$, are always upper bounded by $\gamma_i \leq \mu_{d-n+i}$. This upper bound is saturated if $M$ is diagonal.} \cite{Hwang2004} corroborates that no other choice of basis will exceed the Bell value obtained using the diagonal basis. Hence, it is ensured that the optimal Bell value is achieved for the joint observables in $\mathcal{R}$, and consequently, the optimal Bell value becomes a function of the singular values of the correlation matrix of $\rho$ as 
\begin{equation}
    \mathcal{I}_n(\rho)=2^{n-1}M_n(\rho)=2^{n-1} \max_{\mathcal{S}^{(n)}_l} \sqrt{\sum_{i=1}^n \mu_i}
\end{equation}
\end{proof} 

Up to this point, we have limited our analysis to cases where the given two-qudit state is decomposable into $m$ copies of a two-qubit state as $\rho=\otimes_m\psi$. However, it has recently been pointed out that there exists a class of two-qudit states that cannot be decomposed into such a form \cite{Kraft2018}. Let us now consider this class of two-qudit states, denoted by $\tilde{\rho}\neq \otimes_m \psi$.

An arbitrary two-qudit state $\tilde{\rho}\in\mathscr{L}(\mathscr{H}_A^d \otimes \mathscr{H}_B^d)$ with $d=2^m$, without loss of generality, can be expressed by Eq.~(\ref{Fanoqubit}). However, this choice of basis set $\boldsymbol{\tau} \otimes \boldsymbol{\tau}'$ to express $\tilde{\rho}$ is not unique, unlike the case when we considered the state to be decomposable. To consider an unique basis set for representing a general $\tilde{\rho}$, one can consider its diagonal operator basis, denoted by $\boldsymbol{\frak{t}}\otimes\boldsymbol{\frak{t}}' \equiv \{\frak{t}_k\otimes\frak{t}'_k\}$. This set of bases is obtained by applying two local rotation matrices $R, S \in SO(d^2 -1)$ to diagonalise the correlation matrix $T$ as $T=RT_{D}S^{-1}$, which do not necessarily preserve the algebra of $\boldsymbol{\tau}$, consequently, the basis $\boldsymbol{\frak{t}}$ lacks the properties ensured by Eq.~(\ref{taubasis}).

Now, a necessary and sufficient criterion can be developed only if Alice's and Bob's observables corresponds to the diagonal operator basis, satisfying the constraints $A_{x}^{2}=B_{y}^{2}=\openone$ and  $\{\mathcal{A}_y,\mathcal{A}_{y'}\}=0 \ \forall y \neq y' \in \{1,2,...,n\}$, the condition from Proposition \ref{prop}. However, since rotation operation does not necessarily preserve anticommutation relations or dichotomicity, it is not always possible to map $\mathcal{A}\otimes B$ to the set of basses $\boldsymbol{\frak{t}}\otimes\boldsymbol{\frak{t}}'$, leading to the impracticability of deriving a necessary and sufficient condition.

On the other hand, both the dichotomicity and anticommutation relations are preserved if the transformation is local unitary in $ SU(d)$. For this reason, we resort to those states that are diagonalizable using local unitary operations. 

\begin{cor} \label{cor2}
    For a given two-qudit state $\tilde{\rho}\neq \otimes_m \psi$ of local dimension $d=2^m$, whose $T_{\tilde{\rho}}$ is diagonalizable with local unitary operations, the necessary and sufficient condition for violating the $n$-settings Bell inequality in Eq.~(\ref{nbell}), is given by
    \begin{equation}
        \mathcal{M}_n(\tilde{\rho})=\max_{\mathcal{S}^{(n)}_l} \sqrt{\sum_{k=1}^n \frak{u}_k} > 2^{1-n} \mathcal{I}_n(\tilde{\rho})
    \end{equation}
where $\{\frak{u}_k\}$ are singular values of $T_{\tilde{\rho}}$ corresponding to the diagonal operator basis $\{\frak{t}_k\otimes\frak{t}'_k\}$. 
\end{cor}

\begin{proof}
The proof is similar to that of Theorem \ref{theo} and Corollary \ref{cor1}. The crucial difference in this case is that when the state is not decomposable, the diagonal operator basis will not generally be in the form of Eq.~(\ref{taubasis}). In fact, the diagonal operator basis ($\boldsymbol{\frak{t}}$, $\boldsymbol{\frak{t}}'$) pertaining to each party transforms into an entangled basis. Since we have obtained this basis by applying unitary operations on the basis in Eq. (\ref{taubasis}), the algebra of both sets of bases is the same. Similarly to the case in Corollary \ref{cor1}, if we take Alice's observables as $\mathcal{A}_y=\frak{t}_k$, Bob's observables are fixed as $B_y=\frak{t}'_y$. Eigenvalue interlacing theorem guarantees that the optimal value of the Bell functional will be achieved from the set $\mathcal{R}$, ensuring optimal Bell value is always evaluated from the function $\mathcal{M}_n(\rho)=\max_{\mathcal{S}^{(n)}_l}\sqrt{\sum_k \frak{u}_k}$. 
\end{proof}


We now study the case for $n=3$, for which $M_n (\rho)$ is independent of the choice of basis $\{\sigma_i \otimes \sigma'_j \}$. This is because the set $\mathcal{S}^{(n)}_l$ exhausts the operator basis for a two-qubit state. This enables $M_n (\rho)$ to saturate $\mathcal{I}_n (\rho)$ for all two-qubit states as stated in the following corollary. 
\begin{cor}\label{cor3}
For any given two-qubit state $\rho$, the necessary and sufficient condition for violating the Bell inequality for $n=3$ in Eq.~(\ref{nbell}), is given by
      \begin{equation} \label{3bellm}
          M_3(\rho)= \sqrt{\sum_{i=1}^3 \mu_i} > \frac{2}{3}
      \end{equation}
      where $\{\mu_i\}$ are the eigenvalues of $T_{\rho} \cdot T^{\top}_{\rho}$.
 \end{cor}

\begin{proof}
For a two-qubit state as in Eq.~(\ref{Fanoqubit}), there is only one mutually anticommuting set $\mathcal{S}^{(n)}_l = \{\sigma_x ,\sigma_y ,\sigma_z \}$. The elements of the correlation matrix $T_{\rho}$ are $t_{uv}= \Tr[\rho(\sigma_u \otimes\sigma'_v)]$.  Since any bipartite qubit state can be diagonalized with operators in diagonal positions mutually commuting, Bob's measurements need to be anticommuting. This fact, along with Proposition~\ref{prop}, readily implies that the solution of the optimization problem in Eq.~(\ref{obj}) comes from the set $\mathcal{R}$, ensuring the saturation of Eq.~(\ref{cs2}). 
\end{proof}

\begin{figure}[!ht] 
\includegraphics[scale=0.45]{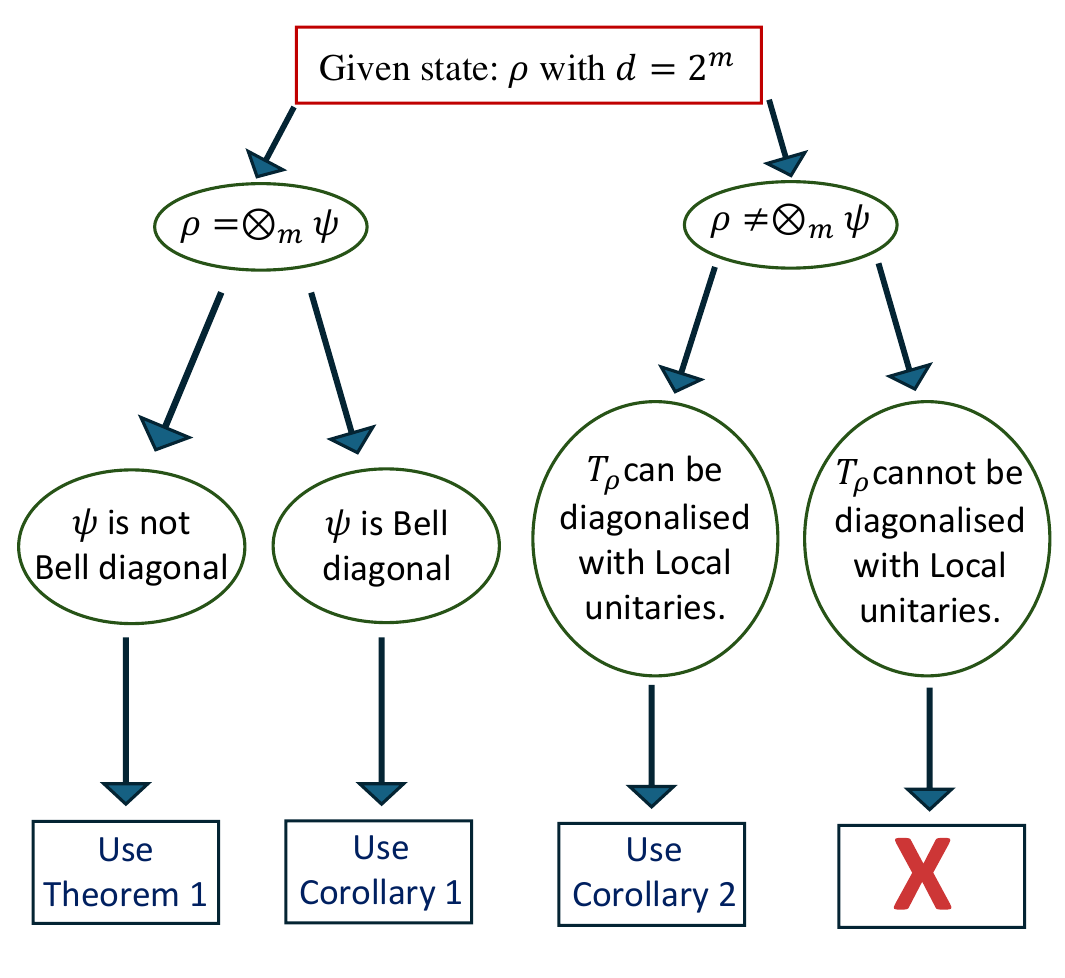}
\caption{A schematic diagram for evaluating the optimal Bell value for a given state.}
\label{picop}
\end{figure}


\subsection{Construction of Observables} \label{conobs}

Now, we provide the construction of the optimal observables for which the optimal bell value is achieved for any given state. We begin by providing a construction of the observables corresponding to the given two-qubit state to saturate the inequality for $n=3$. In this case, the anticommuting sets of operators on $\mathscr{H}^2$ exhaust the basis. We obtain observables of Alice and Bob that saturate the upper bounds as follows
\begin{eqnarray}
    \mathcal{A}_y = \sigma_y ; \ \ 
    B_y = \sigma'_y; \ \ \omega_y = \frac{4 \mu_y}{\sqrt{\mu^2_1 +\mu^2_2 +\mu^2_3}}
\end{eqnarray}

Next, $\mathcal{A}_y = \frac{A'_y}{\omega_y}$ leads to $A'_y =\omega_y \sigma_y$. Using the orthogonal matrix $\mathscr{M}$ as defined in Proposition \ref{prop}, along with a set of conditions on Alice's observables derived in Appx. \ref{parity obliv}, we invert the relationship in Eq.~(\ref{obs3}) to express $A$ in terms of $A'$ as follows
\begin{equation}
\begin{aligned}
    A_1 &= \frac{\mu_1 \sigma_1 +\mu_2 \sigma_2 +\mu_3 \sigma_3}{\sqrt{\mu^2_1 +\mu^2_2 +\mu^2_3}} \  ;  \ A_2 = \frac{\mu_1 \sigma_1 +\mu_2 \sigma_2 -\mu_3 \sigma_3}{\sqrt{\mu^2_1 +\mu^2_2 +\mu^2_3}}  \\
    A_3 &= \frac{\mu_1 \sigma_1 -\mu_2 \sigma_2 +\mu_3 \sigma_3}{\sqrt{\mu^2_1 +\mu^2_2 +\mu^2_3}} \  ;  \ A_4 = \frac{\mu_1 \sigma_1 -\mu_2 \sigma_2 -\mu_3 \sigma_3}{\sqrt{\mu^2_1 +\mu^2_2 +\mu^2_3}} \nonumber
\end{aligned}
\end{equation}
It is straightforward to check that when employing these measurement settings $\{A_x,B_y\}$ the optimal Bell value for an arbitrary given bipartite qubit state will be $4\sqrt{\mu^2_1 +\mu^2_2 +\mu^2_3}$.

Following a similar approach, the observables pertaining to $m$ copies of two-qubit entangled states for attaining the value $M_n (\rho)$ is constructed as
\begin{eqnarray} \label{obsconstruct}
    A_x&=& \frac{1}{\sqrt{\sum_{y,v} t^2_{yv}}} \sum_{y=1}^n \mathscr{M}^{\top}_{xy} \ \qty(\sqrt{\sum_{v} t^2_{yv}} \  \mathcal{A}_y )\nonumber \\
    B_y&=& \sum_{v} \qty(\frac{t_{yv}}{\sqrt{\sum_{v}t_{yv}^{2}}}) \ \tau'_v 
\end{eqnarray}
where $\mathscr{M}^{\top}_{xy}$ denotes the element in the $x^{th}$ row and $y^{th}$ column of the matrix $\mathscr{M}^{\top}$, which represents the inverse map from $\mathcal{A}$ to $A$, satisfying $\mathscr{M}^{\top} \cdot \mathscr{M}= \openone$ as shown in Appx.~\ref{nproof}. The sum over indices $y$ and $v$ run over sets $\{\tau_y \in S_l^{(n)} \}$ and $\{ \tau'_v \in S_{l'}\}$ respectively. It can be readily seen that this choice of observables saturates Eq.~(\ref{newbound}) when employed on the state $\rho$ given by Eq.~(\ref{Fanoqudit}). Then we obtain the corresponding Bell value as
 
\begin{equation}
\begin{aligned}
    \mathcal{I}_n (A,B,\rho) &=\frac{1}{d^2}\sum_{y,v,i,j} \omega_y \frac{t_{yv}}{\sqrt{\sum_v t^2_{yv}}} t_{ij} \Tr [(\tau_i \otimes \tau'_j)(\tau_y \otimes \tau'_v )] \\
    &=\sum_{y,v} \omega_y \delta_{iy}\delta_{jv}\frac{t_{yv} t_{ij}}{\sqrt{\sum_v t^2_{yv}}} \\
    &=\sum_{y,v} \omega_y \frac{t^2_{yv}}{\sqrt{\sum_v t^2_{yv}}} =2^{n-1} \sqrt{\sum_{y,v} t^2_{yv}}
    \end{aligned}
\end{equation}
\normalsize
where the sum over indices $y,v$ run over local dimensions such that $\{\tau_y \in S_l^{(n)} \}$ and $\{ \tau'_v \in S_{l'}\}$ respectively. Note that Eq. (\ref{obsconstruct}) applies to both decomposable and non-decomposable classes of two-qudit entangled states. When the state is decomposable, the operators in $\boldsymbol{\mathcal{A}}$ and $\boldsymbol{\tau}'$ are separable, as considered in Eq. (\ref{taubasis}). However, when the state is not decomposable, these operators are entangled.

Furthermore, for better clarity, we construct the observables pertaining to an arbitrary two copies of Bell diagonal two-qubit state for $n=4$  in Appx.~\ref{example}.

\emph{Remark.--}  The necessary and sufficient criterion hinges on the ability to diagonalize a given bipartite state using local unitary operations. Since any two-qubit state can be diagonalized by using local unitaries, the optimal value $\mathcal{I}_n (\rho)$ can be reached by observables within $\mathcal{R}$ following the arguments presented in Corollaries \ref{cor1} and \ref{cor2}. Furthermore, for the CHSH inequality, the maximum over anticommuting sets $\mathcal{S}_l$ translates to the sum of the two largest eigenvalues, precisely the famous Horodecki criterion \cite{Horodecki1995}.


\section{Summary and Outlook}\label{outlook}

Our analysis reveals that developing a necessary and sufficient criterion for a given bipartite state of local dimension $d=2^m$ to determine if it violates a $n$-settings Bell inequality given by Eq.~(\ref{nbell}) is intricately linked to the diagonalization of the state's correlation matrix. Specifically, this criterion rests on the demonstration that such diagonalization produces singular values which determine the optimal Bell value, with the diagonal operator basis corresponding to the optimal measurement directions for Alice and Bob. 

In our approach, the optimization over all possible observables $A_x$ and $B_y$ results in each joint expectation value $\braket{\mathcal{A}_y \otimes B_y}$ in Eq.~(\ref{obj}) corresponding to one of the diagonal elements of the state's diagonalized correlation matrix (singular values). Consequently, Alice's and Bob's measurement basis are determined by the diagonal operator basis for the state, thereby indicating that singular values are the optimal expectation values of $\mathcal{A}_y \otimes B_y$. 

For two-qubit states, we show that if Alice and Bob perform their measurements in the diagonal operator basis, the Bell value will be optimum for the corresponding state, establishing a connection between the optimal Bell value and the function $M_n(\rho)$. This enables us to come up with a necessary and sufficient Horodecki-like criterion in terms of the state parameters to determine whether the state violates Gisin's elegant Bell inequality ($n=3$) \cite{Gisin2009}, which has found several applications in information processing tasks such as the RAC game \cite{Ghorai2018}, device-independent certification of unsharp measurement \cite{Roy2023} and SIC-POVM \cite{Smania2020}, self-testing \cite{Andersson2017, Chen2021}, randomness certification \cite{Andersson2018} and Network nonlocality \cite{Tavakoli2017}. 

For two-qudit states, the diagonalization of the correlation matrix $T_{\rho}$ is more nuanced compared to that of two-qubit states.  For two-qubit states, rotation in $SO(3)$ is isomorphic to an unitary operation $SU(2)$, ensuring local unitary operations suffice for diagonalization and preserve the algebra of basis operators. However, for two-qudit states, $SU(d)\subset SO(d^2 -1)$, implying that the rotations used for diagonalization may lack local unitary representations and generally do not preserve the necessary properties of the required operator bases. After diagonalization, the operator basis may fail to satisfy the constraints in Eq.~(\ref{obj}), which are essential for our scenario and the considered $n$-settings Bell inequality.

A fundamental challenge in establishing necessary and sufficient conditions for Bell violation stems from the fact that $n$-settings Bell inequalities, as two-outcome inequalities, do not offer adequate variability to develop such criteria for arbitrary two-qudit states. Therefore, while it may not always be possible to derive a necessary and sufficient criterion for any two-qudit state that cannot be diagonalized with local unitary operations, our analysis shows that a sufficient criterion can be obtained by carefully defining $M_n(\rho)$ in a manner that tightly corresponds to those states which can be diagonalized with local unitary operations, thereby becoming necessary and sufficient.

To this end, by utilizing the symmetries in the Bell expression, we prove that the optimal value is reached when Alice's scaled observables defined in Eq.~(\ref{obs3}) are mutually anticommuting in proposition \ref{prop}. Next, by considering a class of states for which the given two-qudit state is decomposable into $m$ copies of the two-qubit state,  and assuming that $B_y$ can be expressed as a linear combination of anticommuting observables, we evaluate a lower bound on the solution of the optimization problem defined in Eq.~(\ref{obj}) in terms of the elements of the correlation matrix. While expressing $B_y$ as a linear combination of anticommuting observables ensures $B_y^2=\openone$, it does not guarantee this is the only way to satisfy the dichotomicity condition.  Thus, under this assumption, the optimization may not saturate the optimal quantum value because we are only optimizing over a subset of dichotomic traceless observables on Bob’s side. Hence, obtaining the optimal value in this manner constitutes a sufficient criterion for violating the $n$-settings Bell inequality for any given two-qudit state of $d=2^m$. 

On the other hand, for the rest of the states, when the given two-qudit state is not decomposable into the form $\otimes_m \psi$, diagonalization of the correlation matrix generally results in a diagonal operator basis that may not satisfy the constraints in Eq.~(\ref{obj}). However, within this class, there always exists another subclass of states whose correlation matrix can be diagonalised using local unitary operations. For this subclass of states, the optimal condition for achieving the optimal Bell value is for Bob's measurements to be anticommuting, and the optimal value is attained within $\mathcal{R}$. This condition then becomes \emph{necessary} as well (see flow-chart~\ref{picop} for an unified illustration of our findings). 


Our work opens up several future directions. Firstly, it calls for the characterization of the class of states, other than the Bell-diagonal ones, that can be diagonalized with local unitaries while preserving the dichotomicity and the requirements of Proposition \ref{prop}. Another direction is to develop a necessary and sufficient criterion, following our approach, for a state to violate $d$-outcome Bell inequality with constraints other than dichotomicity, allowing a larger set of observables to find the optimal value. An example of such a $d$-outcome Bell inequality has been proposed \cite{Salavrakos2017}, which has been utilized in the context of self-testing two-qudit entangled states and randomness generation \cite{Sarkar2021}. 

Moreover, our method can be applied to characterize the class of states useful for the sequential sharing of correlations among an arbitrary number of observers. While a certain class of states has been found useful for such applications, the classification of the entire set remains an open question \cite{Brown2020, Sasmal2023}. To achieve this, one must generalize our criterion by optimizing over all possible sets of POVMs, following the method introduced in \cite{Hall2022}, rather than being confined to projective measurements.


\section{Acknowledgements}  
SB acknowledges the hospitality of IIT Hyderabad and the support of the research grant IITH/SG160 and the INFOSYS grant. SK acknowledges the support from the Digital Horizon Europe project, FoQaCia,(\textit{Foundations of quantum computational advantage}), GA no,202070558, funded by the European Union and NSERC (Canada).  AKP acknowledges the support from the research grant SERB/MTR/2021/000908, Government of India. SS acknowledges the support from the National Natural Science Fund of China (Grant No. G0512250610191).


 

%


\appendix

\onecolumngrid

\section{Detailed proof of Proposition \ref{prop}} \label{parity obliv}

We first identify some essential features of the proof from the case $n=3$ in order to generalize for arbitrary $n$.

\subsection{Proof of $\max \sqrt{\omega_1^2 +\omega_2^2 + \omega_3^2} = 4$ when $\{\mathcal{A}_y,\mathcal{A}_{y'}\}=0 \ \forall y\neq y'\in\{1,2,3\}$}

Recalling the definition of $\omega_y$ given by Eq.~(\ref{obs3}) of main text, we evaluate the following
\begin{equation}
    \max\limits_{\{A_x\}} \sqrt{\sum_{y=1}^3 \omega_y^2} = \max\limits_{\{A_x\}} \sqrt{\Tr[\rho\qty({A_1^{'}}^2 +{A_2^{'}}^2 + {A_3^{'}}^2 )]} \ ; \ \ \ A'_y = \sum\limits_{x=1}^4 (-1)^{z_y^x} A_x \ \forall y \in [3]
\end{equation}
In order to evaluate $\sum\limits_{y=1}^3{A_y^{'}}^2$, we represent the coefficients $(-1)^{z^x_y}$ as a $3 \times 4 $ matrix $M$. We see that all the rows of $M$ are orthogonal to each other. We construct an orthogonal matrix $\mathscr{M}$ by adding a row $(-1, 1, 1, 1)$ to the matrix $M$.
\begin{eqnarray}
    M = \begin{bmatrix}
        1 & 1 & 1 & 1 \\
        1 & 1 & -1 & -1 \\
        1 & -1 & 1 & -1
    \end{bmatrix} \ \  \longrightarrow \ \  
    \mathscr{M} =\begin{bmatrix}
        1 & 1 & 1 & 1 \\
        1 & 1 & -1 & -1 \\
        1 & -1 & 1 & -1 \\
        1 & -1 & -1 & 1
    \end{bmatrix}
\end{eqnarray}
We introduce a dummy unnormalised observable $A'_4$ to the set $\qty{A'_y}_{y=1}^{3}$ such that the following matrix equation holds
\begin{eqnarray} \label{matdef}
    \begin{bmatrix}
        A'_1 \\
        A'_2 \\
        A'_3 \\
        A'_4
    \end{bmatrix}=
    \begin{bmatrix}
        1 & 1 & 1 & 1 \\
        1 & 1 & -1 & -1 \\
        1 & -1 & 1 & -1 \\
        1 & -1 & -1 & 1
    \end{bmatrix} \
    \begin{bmatrix}
        A_1 \\
        A_2 \\
        A_3 \\
        A_4
    \end{bmatrix}
\end{eqnarray}
Taking the transpose of Eq.~(\ref{matdef}) and multiplying with itself, we get
\begin{eqnarray}
\begin{bmatrix}
        A'_1 & A'_2 & A'_3 & A'_4
    \end{bmatrix}
    \begin{bmatrix}
        A'_1 \\
        A'_2 \\
        A'_3 \\
        A'_4
    \end{bmatrix}&=&
    \begin{bmatrix}
        A_1 & A_2 & A_3 & A_4
    \end{bmatrix}
    \begin{bmatrix}
        1 & 1 & 1 & 1 \\
        1 & 1 & -1 & -1 \\
        1 & -1 & 1 & -1 \\
        1 & -1 & -1 & 1
    \end{bmatrix}
    \begin{bmatrix}
        1 & 1 & 1 & 1 \\
        1 & 1 & -1 & -1 \\
        1 & -1 & 1 & -1 \\
        1 & -1 & -1 & 1
    \end{bmatrix}
    \begin{bmatrix}
        A_1 \\
        A_2 \\
        A_3 \\
        A_4
    \end{bmatrix} \nonumber\\
\Rightarrow \ \ {A_1^{'}}^2 +{A_2^{'}}^2 + {A_3^{'}}^2 + {A_4^{'}}^2 &=& 4 \ \qty(A_1^2 +A_2^2 +A_3^2+A_4^2) \nonumber\\
\Rightarrow \ \ {A_1^{'}}^2 +{A_2^{'}}^2 + {A_3^{'}}^2 &=& 16 \ \openone -{A_4^{'}}^2 \ \ \ \  [\text{Since} \ A_i^2 =\openone] \label{n31}
\end{eqnarray}
Hence, from Eq.~(\ref{n31}), we obtain $\max\limits_{\{A_x\}} \sqrt{\Tr[\rho\qty({A_1^{'}}^2 +{A_2^{'}}^2 + {A_3^{'}}^2 )]} = \max\limits_{\{A_x\}} \sqrt{16-\Tr[\rho \ {A_4^{'}}^2]}$. Since ${A_4^{'}}^2$ and $\rho$ are positive semi-definite, $\Tr[\rho \ {A_4^{'}}^2] \geq 0$. Thus, for optimal value of $\sum\limits_{y=1}^3{A_y^{'}}^2$, the corresponding Alice's observable leads to $\Tr[\rho \ {A_4^{'}}^2] =0$ for any given $\rho$, implying $A_4^{'}=0$. Hence, $ \max\limits_{\{A_x\}} \sqrt{\sum_{y=1}^3 \omega_y^2} = 4$. We derive the constraints on Alice's observables from the operator relation ${A_4^{'}}=0$. Putting this in Eq.~(\ref{matdef}) and multiplying by $\mathscr{M}^{-1}$ on both sides, we obtain
\begin{eqnarray}
\frac{1}{4}\begin{bmatrix}
        1 & 1 & 1 & 1 \\
        1 & 1 & -1 & -1 \\
        1 & -1 & 1 & -1 \\
        1 & -1 & -1 & 1
    \end{bmatrix}
    \begin{bmatrix}
        A'_1 \\
        A'_2 \\
        A'_3 \\
        0
    \end{bmatrix}=
    \begin{bmatrix}
        A_1 \\
        A_2 \\
        A_3 \\
        A_4
    \end{bmatrix} 
\end{eqnarray}
which gives us $\{A_x\}$ in terms of $A_y^{'}$ as
\begin{eqnarray}\label{ai}
 A_1 =  \frac{1}{4} (A_1^{'} + A_2^{'} + A_3^{'}) \ ; \ \ A_2=\frac{1}{4} (A_1^{'} + A_2^{'} - A_3^{'}) \ ; \ \ A_3= \frac{1}{4} (A_1^{'} - A_2^{'} + A_3^{'}) \ ; \ \ A_4 = \frac{1}{4} (A_1^{'} - A_2^{'} - A_3^{'})
\end{eqnarray}
From these equations, we find $\{A_y^{'},A_{y'}^{'}\}\equiv A'_{yy'}$ by using $A_{x}^2 =A_{x'}^2 ~ \forall x,x'$. Since there are three anticommutation relations, we can always suitably choose three equations out of six available equations to solve for $A'_{yy'}$. By choosing, $A_1^2 =A_x^2$ we obtain the following
\begin{equation}
    A'_{13} +A'_{23}=0  \ ; \ \ A'_{12} +A'_{23}=0 \ ; \ \ A'_{12} +A'_{13}=0 \label{acn3}
\end{equation}
Eq.~(\ref{acn3}) can be solved by considering the following matrix equation
\begin{eqnarray}
\label{aceqn3}
    \begin{bmatrix}
        0 & 1 & 1 \\
        1 & 0 & 1 \\
        1 & 1 & 0
    \end{bmatrix}
    \begin{bmatrix}
        A'_{12} \\
        A'_{13} \\
        A'_{23} 
    \end{bmatrix}=0
\end{eqnarray}
This is a linear homogeneous equation of three variables. Since each row of the square matrix is linearly independent, it is a rank-3 matrix. Therefore, the only solution of this set of equations is that $\{A'_{y} ,A'_{y'}\} =0 \ \forall y,y'$. 

For a general $n$, we use the following features of the preceding example:
\begin{enumerate} [(i)]
\item Starting from the relation between $A'_y$ and $A_x$, we represent it as a linear transformation by a $(n \times 2^{n-1} )$ matrix $M$ on a column vector whose elements are $\{A_x\}$.
\item  We then show that for any $n$ the rows of $M$ are orthogonal. Using this, we can add rows to $M$ to make it a $(2^{n-1} \times 2^{n-1})$ orthogonal matrix $\mathscr{M}$. We can rewrite the relation between $A'_y$ and $A_x$ using the matrix $\mathscr{M}$ by adding dummy observables $\{A'_l\}_{l=n+1}^{2^{n-1}}$.
\item Using the property $A_x^2 =\openone$ and the structure of $\mathscr{M}$ we rewrite $\sqrt{\sum_{y=1}^n \omega_y^2}$ in terms of $\sum_{l=n+1}^{2^{n-1}}A'^2_l$, which consists of the dummy operators. The maximum of $\sqrt{\sum_{y=1}^n \omega_y^2}$ is attained when $A'_l =0 \ \forall l>n$.
\item We then invert the relation obtained in step 3 and substitute $A_x$ in terms of $A'_y$ in the relation $A_x^2 = A_{x'}^2$ to obtain a set of equations in terms of the anticommutators $\{A'_y, A'_{y'} \}$.
\item We prove that this set of equations has no nontrivial solution which leads to the anticommutation relations $\{A'_y, A'_{y'} \}=0 \ \forall y,y'\in [n]$. In other words, the generalization of the matrix defined in Eq.~(\ref{aceqn3}) with dimension $\frac{n(n-1)}{2}\times(2^{n-1}-1)$ contains $\frac{n(n-1)}{2}$ linear independent rows. This results in only trivial solutions for the anticommutation relations (see Eq.~\ref{aceqn}).
\end{enumerate}


\subsection{Proof of $\max \sqrt{\sum_{y=1}^n \omega_y^2} = 2^{n-1}$ when $\{\mathcal{A}_y,\mathcal{A}_{y'}\}=0 $ for arbitrary $n$} \label{nproof}

Similar to the case for $n=3$, in order to evaluate $ \max\limits_{\{A_x\}} \sqrt{\sum_{y=1}^n \omega_y^2} = \max\limits_{\{A_x\}} \sqrt{\Tr[\rho\qty(\sum\limits_{y=1}^n{A_y^{'}}^2)]}$, where $A'_y = \sum\limits_{x=1}^{2^{n-1}} (-1)^{z_y^x} A_x \ \forall y \in [n]$, we first consider a $n \times 2^{n-1}$ matrix $M$ where element of the $y^{th}$ row of $x^{th}$ column is given by $M_{yx}=(-1)^{z_y^x} \ \forall x \in [2^{n-1}], y \in [n]$.  The term ${z_{y}^{x}}$ is the $(n+1-y)^{th}$ bit of the number $(x-1)$ expressed in binary - or equivalently they can be defined as $n$ bit strings whose first bit is zero and the rest $n-1$ bits exhaust all possible strings of length $n-1$. The dot product between any two rows labelled by $(y,y')$, is given by $\sum_x M_{yx} M_{y'x}=\sum_x (-1)^{z^x_y \oplus z^x_{y'}}$.  The structure of $M$ is such that there are $2^{n-2}$ elements satisfying $(-1)^{z^x_y \oplus z^x_{y'}}=1$ and the other $2^{n-2}$ elements have values $(-1)^{z^x_y \oplus z^x_{y'}}=-1$. For any two rows $y,y'$, $z^x_y =z^x_{y'}$ for $2^{n-2}$ values of $x$ and $z^x_y \neq z^x_{y'}$ for the remaining $2^{n-2}$ values. This can be shown in the following way. The elements of $z^x_y$ are defined as $n$ bit strings whose first bit is zero, and the rest $n-1$ bits exhaust all possible strings of length $n-1$. Lets fix the value of $y^{th}$ and $y'^{th}$ bits, say $z^x_y =z^{x}_{y'}=0$. There are $2^{n-3}$ strings such that $z^x_y =z^{x}_{y'}=0$. If we now impose the condition $z^x_y =z^x_{y'}=0,1$ on the $y^{th}$ and $y'^{th}$ bit, we will get $2^{n-2}$ such strings. The remaining strings satisfy $z^x_y \neq z^x_{y'}$.  Thus $\sum_x M_{yx} M_{y'x}=0$, implying any two rows of $M$ are orthogonal.

Now, we can always add $(2^{n-1} -n)$ numbers of mutually orthogonal rows to $M$ so that it becomes an orthogonal matrix. Let the newly constructed orthogonal matrix be denoted by $\mathscr{M}$. An element in the $y^{th}$ row of $x^{th}$ column is denoted by $\mathscr{M}_{yx}: \forall y,x \in [2^{n-1}]$. Since $\mathscr{M}$ is orthogonal, we have $\mathscr{M}^{\top} \cdot \mathscr{M} = 2^{n-1}\openone$.

Next, we map the set of equations $A'_y = \sum\limits_{x=1}^{2^{n-1}} (-1)^{z_y^x} A_x \ \forall y \in [n]$ into a matrix equation by introducing two column matrix $\textbf{A}'$ and $\textbf{A}$, each having $2^{n-1}$ elements. Each element of $\textbf{A}$ represents one of Alice's observables $\{A_x\}_{x=1}^{2^{n-1}}$. First $n$ elements of $\textbf{A}'$ corresponds to linear combinations of $\{A_x\}_{x=1}^{2^{n-1}}$ that appear in the Bell inequality. The rest are dummy observables $\{A'_y\}_{y=n+1}^{2^{n-1}}$. Then we represent the set of equations in terms of the following matrix equation
    \begin{equation}
        \textbf{A}' = \mathscr{M} \cdot \textbf{A} \label{matdef1}
    \end{equation}
Taking the transpose of Eq.~(\ref{matdef1}) and multiplying with itself, we get
\begin{eqnarray}
     \textbf{A}'^{\top} \cdot \textbf{A}' &=& \textbf{A}^{\top}  \mathscr{M}^{\top} \mathscr{M}  \ \textbf{A} = 2^{n-1} \textbf{A}^{\top} \cdot \textbf{A} \nonumber \\
     \implies \sum\limits_{y=1}^{2^{n-1}} {A^{'}_y}^2 &=& 2^{n-1} \sum\limits_{x=1}^{2^{n-1}} {A^2_x} = 4^{n-1} \ \openone \nonumber \\
     \implies  \sum\limits_{y=1}^{n} {A^{'}_y}^2 &=& 4^{n-1} \ \openone - \sum\limits_{y=n+1}^{2^{n-1}} {A^{'}_y}^2 \label{trans} \\
     \implies  \max\limits_{\{A_x\}} \sqrt{\Tr[\rho\qty(\sum\limits_{y=1}^n{A_y^{'}}^2)]} &=& \max\limits_{\{A_x\}} \sqrt{4^{n-1} - \Tr[\rho\qty(\sum\limits_{y=n+1}^{2^{n-1}}{A_y^{'}}^2)] } \label{trans2}
\end{eqnarray}
Note that since any given $\rho$ and $A'^2_y$ are positive semi-definite, $\Tr[\rho\qty(\sum\limits_{y=n+1}^{2^{n-1}}{A_y^{'}}^2)]\geq 0 \ , y \in \{n+1, n+2,...,2^{n-1}\}$. Hence, the optimal value for each given $\rho$ occurs if $\sum\limits_{y=n+1}^{2^{n-1}}\Tr[\rho\qty({A_y^{'}}^2)] =0 $, implying $\Tr[\rho A'_y]=0 \ \forall \rho,y \in \{n+1, n+2,...,2^{n-1}\}$ that leads to $(2^{n-1}-n)$ state independent constraints on Alice's observables. This conforms with the parity oblivious conditions \cite{Ghorai2018} given by $A'_y=0$ when $y \in \{n+1, n+2,...,2^{n-1}\}$. Using this condition, from Eq.~(\ref{trans2}), the maximum value is derived as
\begin{equation}\label{alicecondn}
    \max \sqrt{\sum_{y=1}^n \omega_y^2} = 2^{n-1} \ \ \text{when $A'_y=0 \ \forall y \in \{n+1,n+2,...,2^{n-1}\}$ }
\end{equation}
Now, we derive that the relation $A'_y=0 \ \forall y \in \{n+1, n+2,...,2^{n-1}\}$ leads to the anticommutation relation, $\{A'_y ,A'_{y'}\}=0 ~\forall y\neq y' \in [n]$. Multiplying $\mathscr{M}^{\top}$ to Eq.~(\ref{matdef1}), we  obtain $\textbf{A}= 2^{1-n}\mathscr{M}^{\top} \cdot \textbf{A}'$ since $\mathscr{M}^{\top}\mathscr{M}=2^{n-1}\openone$, implying the inverse relationship between $\{A_x\}_{x=1}^{2^{n-1}}$ and $\{A'_y\}_{y=1}^{n}$, given by $ A_x = 2^{1-n}\sum_{y=1}^{2^{n-1}} (-1)^{z^x_y} A'_y = 2^{1-n} \sum_{y=1}^{n} (-1)^{z^x_y} A'_y$ since $A'_y =0~\forall y> n$. Using the identity $A^2_1 =A^2_{x} ~\forall x >1$, we obtain the following set of $(2^{n-1} -1)$ equations
 \begin{equation}\label{aceqn}
 \begin{aligned}
 &\sum_{y,y'=1}^n A'_y A'_{y'} = \sum_{y,y'=1}^n (-1)^{z^x_y \oplus z^x_{y'}}A'_y A'_{y'} \ \ \ \ \ [\text{for each} \ x>1] \\
     \implies & \sum_{y< y'} \qty[1-(-1)^{z^x_y \oplus z^x_{y'}}] \ \qty{A'_y,A'_{y'}} =0 \\
     \implies & \sum_{y<y'} \qty(z^x_y \oplus z^x_{y'}) \ \qty{A'_y,A'_{y'}} =0  \\
     \implies & \textbf{P} \cdot \mathcal{S}  =0   
     \end{aligned}
 \end{equation}
where the row matrix $\textbf{P}\equiv \qty[\{A'_1,A'_{2}\} \ \{A'_1,A'_{3}\} \ \cdots \{A'_{(n-1)},A'_{n}\}]$ contains $\frac{n(n-1)}{2}$ number of elements represents each anticommutation relation. $\mathcal{S}$ is a $\qty(\frac{n(n-1)}{2}) \times \qty(2^{n-1}-1)$ matrix. We label rows of matrix $\mathcal{S}$ by $\mu \equiv (y,y'); \ y'>y$ and the columns by $j$. Each element is given by
\begin{equation}\label{smat}
\mathcal{S}_{\mu j}\equiv z^{(j+1)}_y \oplus z^{(j+1)}_{y'} \ \ \text{with $x=j
+1$}
\end{equation}
Note that $\text{rank}(\mathcal{S}) \leq  \frac{n(n-1)}{2} \leq (2^{n-1}-1) \ \ \forall n \geq 2$. Now, we show that $\mathcal{S}$ is full rank with $\text{rank}(\mathcal{S}) =  \frac{n(n-1)}{2}$, which immediately implies that the only possible solution of Eq.~(\ref{aceqn}) is $\qty{A'_y ,A'_{y'} }=0 \ \forall y < y' \in [n]$.

To begin with, we explain how we transform $\mathcal{S}$ into an upper triangular form by suitably swapping columns and performing a sequence of row operations. Since row and column operations are rank preserving, the rank of the transformed matrix is the same as that of $\mathcal{S}$. Now, let us consider the first $(n-1)$ elements of $j^{th}$ column. These elements are labelled by $\mu= \{(1,2),(1,3),\cdots, (1,n)\}$. Using the property $z^x_y =0$ when $y=1$, we obtain
\begin{eqnarray}\label{sblock}
    &&\mathcal{S}_{\mu j} =z^{(j+1)}_1 \oplus z^{(j+1)}_{y'}=z^{(j+1)}_{y'}  \ \ \text{when} \ y=1, y'\in(2,3,...,n) , j\in(1,..., 2^{n-1} -1)  
\end{eqnarray}
Each column represents the binary form of an integer lying between $1$ and $( 2^{n-1} -1)$. Equivalently, they represent all possible bit strings of length $(n-1)$ excluding the bit string with all bits equal to zero corresponding to $x=1$. This is represented in Fig.~\ref{co}(a). 

Now, we rearrange the columns of $\mathcal{S}$ in such a way so the set $\{1,2,\cdots , 2^{n-1} -1 \}$ is partitioned into $n-1$ blocks and the $k^{th}$ block is identified as the collection of bit strings of length $(n-1)$ containing $k$ number of ones. Concretely, the first block of dimension $(n-1)\times(n-1)$ becomes the identity matrix, $\openone_{(n-1)\times(n-1)}$ denoted as \textbf{I} in Fig.~\ref{co}(b). Each column of the next $(n-1)\times \frac{(n-1)(n-2)}{2}$ dimensional block starting from $j=n$  contains two $1$s along with remaining $(n-3)$ zeroes. This is the block labeled \textbf{II} in Fig.~\ref{co}(b). The block \textbf{II} can be further divided into sub-blocks as follows
\begin{equation}\label{smat2}
\textbf{II}=\begin{pmatrix}
\ S^{1} & S^{2} & \ldots & S^{l} & \ldots & S^{n-2}\\
\end{pmatrix}
\end{equation}
The dimension of each sub-block $S^{l}$ is $(n-1)\times(n-l-1)$. Taking $\alpha$ as the row index and $\beta$ as the column index  of $S^{l}$, the corresponding elements are given by
\begin{equation}
    S^{l}_{\alpha\beta} = \begin{cases}
    1 \ \text{ if $\alpha=l$ }\\
    1 \ \text{ if $\alpha=l+\beta$}\\
    0 \ \text{ otherwise}
    \end{cases}
\end{equation}
Here, $1 \leq \alpha \leq n-1$ and $1 \leq \beta \leq n-l-1$. In particular, the column $\beta$ of $S^{l}$ can be written as
\begin{equation}
\label{slbta}
S^{l}_{\beta}=\left(
\begin{array}{c}
\bold{0}_{(l-1)\times 1}\\
1 \\
\bold{0}_{(\beta-1)\times 1}\\
1\\
\bold{0}_{w\times 1}
\end{array}
\right)_{(n-1)\times 1}
\end{equation}
with $w=(n-1)-(l+\beta)$, and $\bold{0}_{v\times 1}$ is a column vector consisting of zeros with $v$ rows.

Subsequently, the $k^{th}$ block containing $k$ number of ones\footnote{Not shown in the figure and not important for our case as explained later.} in each column has dimensions $(n-1)\times \Comb{n-1}{k}$. The $k^{th}$ block starts from the $j^{th}$ column of $\mathcal{S}$ where $j= \sum_{i=1}^{k-1} \Comb{n-1}{i}$. The block labelled as \textbf{V} in Fig. \ref{co}(b) consisting of the rows given by $\mu= (y,y'): (y\geq 2)(y'>y)$ constitutes a submatrix of $\mathcal{S}$ of dimensions $\frac{(n-1)(n-2)}{2} \times (2^{n-1}-1)$. The elements of this and the remaining block in Fig.~\ref{co}(b) can be determined using Eq.~(\ref{smat}). Now we show that the largest square matrix contained in $\mathcal{S}$ (see, Eq.~(\ref{lsqm}) contains linear independent rows implying only trivial solutions exist for Eq.~(\ref{aceqn}) (i.e. solutions of the form $\{A_{y}',A'_{y'}\}=0$).

\begin{figure}[!ht] 
\includegraphics[scale=0.50]{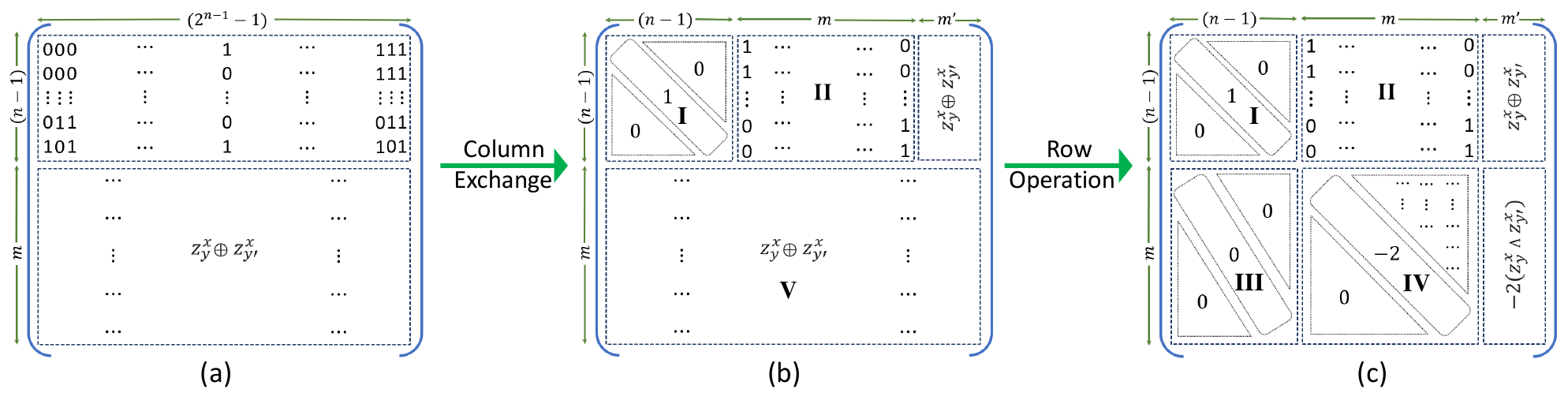}
\caption{The figure depicting the column and row operations on the matrix $\mathcal{S}$ with $m=\frac{(n-1)(n-2)}{2}$ and $m'=2^{n-1}-1-\frac{n(n-1)}{2}$.}
\label{co}
\end{figure}

Next, we perform the following row operations on the block \textbf{III} ($y\geq 2, y'>y, x\leq \frac{n(n-1)}{2}$)
\begin{eqnarray} \label{rowreduc}
   \mathcal{S}_{\mu j}= \mathcal{S}_{(y,y')j} &\longrightarrow& \mathcal{S}_{(y,y')j} - \mathcal{S}_{(1,y)j} -\mathcal{S}_{(1,y')j} \nonumber\\
    &=& \qty[z^{(j+1)}_y \oplus z^{(j+1)}_{y'} ] - \qty[z^{(j+1)}_y + z^{(j+1)}_{y'} ] \nonumber\\
    &=& -2 \ \qty[z^{(j+1)}_y \cdot z^{(j+1)}_{y'}] \nonumber\\
    &=& -2 \ \qty[\mathcal{S}_{(1,y)j} \cdot \mathcal{S}_{(1,y')j}]
\end{eqnarray}
The first equality is implied by using Eq.~(\ref{smat}) and Eq.~(\ref{sblock}). The second equality is the consequence of the fact that each element $z^{x}_{y}$ is binary. We will show that after performing these operations, the matrix $\mathcal{S}$ is in upper triangular form. This will be seen to have $\frac{n(n-1)}{2}$ linear independent rows which implies that the matrix $\mathcal{S}$ is nonsingular and no non-trivial solution to Eq.~(\ref{aceqn}) exists. 

 Recall that the first $(n-1)\times(n-1)$ dimensional block of the matrix in Fig. \ref{co}(b) is the identity matrix. After the row operations, the $m\times (n-1)$ (here, $m= (n-1)(n-2)/2$) block \textbf{III} will contain only zeros. This is implied by Eq.~(\ref{rowreduc}) and the fact that each column in block \textbf{I} has only one non-zero element. Next, we show that $m\times m$ block \textbf{IV} underneath \textbf{II} is an Identity matrix after the row operations. Recalling Eq.~(\ref{smat2}), the $\frac{n(n-1)}{2}\times\frac{n(n-1)}{2}$ block of the matrix $\mathcal{S}$, $\mathcal{S}_{sq}$ which contains the four sub-blocks \textbf{I}, \textbf{II}, \textbf{III}, \textbf{IV} in Fig.~\ref{co}$(c)$ is now written as
\begin{equation} \label{lsqm}
\mathcal{S}_{sq}=\begin{bmatrix}
 \openone_{(n-1)\times(n-1)} & \smash[b]{\overbrace{ \begin{matrix} S^{1} & S^{2} & \ldots & S^{l} & \ldots & S^{n-2} \end{matrix}}^{\textbf{II}}} \\
\bold{0}_{m\times(n-1)}  &  \smash[b]{\underbrace{ \begin{matrix} K^{1} & K^{2} & \ldots & K^{l} & \ldots & K^{n-2} \end{matrix}}_{\textbf{IV}}}
  \end{bmatrix}
\end{equation}
Here $\textbf{I}=\openone_{(n-1)\times(n-1)}$ and $\textbf{III}=\bold{0}_{m\times(n-1)}$. Recall that our goal was to show that the rows of $\mathcal{S}$ are linearly independent. This shows that the square matrix $\mathcal{S}_{sq}$ is nonsingular. In the following, we will show that the row operations in Eq.~(\ref{rowreduc}) reduce the $m\times m$ submatrix \textbf{IV} to $\openone_{m\times m}$. This will directly imply that $\mathcal{S}_{sq}$ is an upper triangular matrix and hence nonsingular. This is a sufficient condition for $\mathcal{S}$ to be linearly independent and we need not consider the remaining columns.  To begin, note that the row operations in Eq.~(\ref{rowreduc}) only change the values in $K^{l}$. Furthermore, Eq.~(\ref{slbta}) implies that $\beta^{th}$ column of $S^{l}$ contains two ones. Thus, the row operation of Eq.~(\ref{rowreduc}) demands that there will be one non-zero element in the $\beta^{th}$th column in $K^{l}$, $K^{l}_{\beta}$. In what follows, we will show that the non-zero element of $K^{l}_{\beta}$ will lie on the diagonal of $\mathcal{S}_{sq}$. To do this, we will find out the row and column positions of the non-zero element of $K^{l}_{\beta}$ with respect to $\mathcal{S}_{sq}$ using Eqs.~(\ref{slbta}) and (\ref{rowreduc}). Note that, first $\bold{0}_{(l-1)\times 1}$ vector gives $0$ for all the $\sum_{k=1}^{l-1}(n-k-1)$ elements of $K^{l}_{\beta}$. Then the $1$ gives $0$ for $(\beta-1)$ elements and then gives $1$. The remaining positions will be filled with $0$s. Therefore, the structure of $K^{l}_{\beta}$ given by
 \begin{equation}
\label{klbta}
K^{l}_{\beta}=\left(
\begin{array}{c}
\bold{0}_{q\times 1}\\
1\\
\bold{0}_{r\times 1}
\end{array}
\right)_{\qty(\frac{n(n-1)}{2})\times 1}
\end{equation}
where $q=\sum_{k=1}^{l-1}(n-k-1) + \beta -1= \frac{(l-1)(2n-l)}{2} +\beta -1$ \footnote{Here $\sum_{k= *}$ acts only on the term $(n-k-1)$.} and $r= \frac{1}{2} [(n-1)(n-2)-(l-1)(2n-l)-2\beta] $.
 
Thus, in the column vector $K^{l}_{\beta}$ the $1$ occurs in the $(q+1)$th row. Adding the previous $(n-1)$ rows of $S^{l}_{\beta}$, we can conclude - in the square matrix $\mathcal{S}_{sq}$, the column vector $(S^{l}_{\beta},K^{l}_{\beta})$ of dimension $n(n-1)/2\times 1$ has $1$ in the row indexed as $(q+n)$. Now, we find the column position in the square matrix $\mathcal{S}_{sq}$ corresponding to the $\beta$th column in $S^{l}(K^{l})$. If this column number is equal to the row number $(q+n)$ obtained above, the one in $K^{l}_{\beta}$ is in the diagonal of $\mathcal{S}_{sq}$.

Recall that, the $\beta$ subscript in $S^{l}_{\beta}(K^{l}_{\beta})$ denotes the $\beta^{th}$ column of the sub-block $S^{l}(K^{l})$ i.e.  $1\leq \beta\leq (n-l-1)$.  Note that, there is an Identity matrix with $(n-1)$ columns. Subsequently, there are $(S^{1},S^{2},\ldots,S^{l-1})$ sub-matrices before $S^{l}$. Since the dimension of each sub-matrix $S^{k}$ is $(n-1)\times(n-k-1)$. the total number of columns of $(S^{1},S^{2},\ldots,S^{l-1})$ will be $\sum_{k=1}^{l-1}(n-k-1)$, adding the first $(n-1)$ columns of the identity matrix and the $(\beta-1)$ columns in $S^{l}$ before $\beta$th column the total number of columns before $\beta$th column of $S^{l}$ is $\sum_{k=1}^{l-1}(n-k-1)+(n-1)+\beta-1=q+n-1$. Therefore, $\beta$th column in $S^{l}$ is the $(q+n)$th column of the matrix $\mathcal{S}_{sq}$. Hence, for all the columns of $S^{l}$, only those elements in $K^{l}$ will be nonzero, which lies on the diagonal of $\mathcal{S}_{sq}$. Thus, we have $(\mathcal{S}_{sq})_{\gamma,\gamma}=1$ where $\gamma=(q+n)$. Thus, the $m\times m$ block \textbf{IV} is an identity matrix.

Therefore, after suitable row and column operations, the matrix $\mathcal{S}$ becomes upper-triangular, implying it is of full rank. This indicates that the only possible solution to Eq.~(\ref{aceqn}) is $\textbf{P}=0$ implying $\{\mathcal{A}_y,\mathcal{A}_{y'}\}=0 $.

For better clarity, we provide an instance of the techniques employed in this proof for the case $n=4$. Specifically, in the following, we restate that $\max \sqrt{\sum_{y=1}^4 \omega_y^2} = 8$ when $\{\mathcal{A}_y,\mathcal{A}_{y'}\}=0 
 \ \forall y \in \{1,2,3,4\}$.

\subsubsection*{An illustration of the general proof for $n=4$}

Recalling Eq.~(\ref{matdef1}) for $n=4$
\begin{eqnarray}\label{matdef4}
   \textbf{A}' = \mathscr{M} \cdot \textbf{A} \implies  
   \begin{bmatrix}
        A'_1 \\
        A'_2 \\
        A'_3 \\
        A'_4 \\
        A'_5   \\
        A'_6 \\
        A'_7 \\
        A'_8
    \end{bmatrix}
    =\begin{bmatrix}
        1 & 1 & 1 & 1 & 1 & 1 & 1 & 1 \\
        1 & 1 & 1 & 1 & -1 & -1 & -1 & -1 \\
        1 & 1 & -1 & -1 & 1 & 1 & -1 & -1 \\
        1 & -1 & 1 & -1 & 1 & -1 & 1 & -1 \\
        1 & 1 & -1 & -1 & -1 & -1 & 1 & 1 \\
        1 & -1 & 1 & -1 & -1 & 1 & -1 & 1 \\
        1 & -1 & -1 & 1 & 1 & -1 & -1 & 1 \\
        1 & -1 & -1 & 1 & -1 & 1 & 1 & -1 
    \end{bmatrix}
    \begin{bmatrix}
        A_1 \\
        A_2 \\
        A_3 \\
        A_4 \\
        A_5 \\
        A_6 \\
        A_7 \\
        A_8 
    \end{bmatrix}
\end{eqnarray}
Expanding Eq.~(\ref{trans}) for $n=4$, we obtain
\begin{eqnarray}
     \sum_{y=1}^4 \langle {A^{'}_y}^2 \rangle &=& 64 - \langle (A_1 +A_2 -A_3 -A_4 -A_5 -A_6 +A_7 +A_8 )^2\rangle -\langle (A_1 -A_2 +A_3 -A_4 -A_5 +A_6 -A_7 +A_8 )^2\rangle \nonumber\\
     && - \langle (A_1 -A_2 -A_3 +A_4 +A_5 -A_6 -A_7 +A_8 )^2\rangle -\langle (A_1 -A_2 -A_3 +A_4 -A_5 +A_6 +A_7 -A_8 )^2\rangle \nonumber\\
     &=& 64 - \sum_{y=5}^8 \langle {A'}_y^2 \rangle
 \end{eqnarray}
Following the argument presented for the general case, for the optimal value $A'_y=0 \ \forall y\in\{5,6,7,8\}$. Now, inverting Eq.~(\ref{matdef4}) yields $\{A_x\}$ in terms of $\{A'_y\}$ as follows (inversion is possible because $\mathscr{M}$ is an orthogonal matrix)
\begin{eqnarray}
 \textbf{A} = \mathscr{M}^{\top} \cdot \textbf{A}' \implies 
    \begin{bmatrix}
        A_1 \\
        A_2 \\
        A_3 \\
        A_4 \\
        A_5 \\
        A_6 \\
        A_7 \\
        A_8 
    \end{bmatrix}
    = \frac{1}{8} \begin{bmatrix}
      1 & 1 & 1 & ~1 & ~1 & ~1 & ~1 & ~1 \\
      1 & 1 & 1 & -1 & ~1 & -1 & -1 & -1 \\
      1 & 1 & -1 & ~1 & -1 & ~1 & -1 & -1 \\
      1 & 1 & -1 & -1 & -1 & -1 & ~1 & 1 \\
      1 & -1 & 1 & ~1 & -1 & -1 & 1 & -1 \\
      1 & -1 & 1 & -1 & -1 & 1 & -1 & ~1 \\
      1 & -1 & -1 & 1 & 1 & -1 & -1 & ~1 \\
      1 & -1 & -1 & -1 & ~1 & 1 & 1 & -1 
    \end{bmatrix}
    \begin{bmatrix}
        A'_1 \\
        A'_2 \\
        A'_3 \\
        A'_4 \\
        0   \\
        0 \\
        0 \\
        0
    \end{bmatrix}
\end{eqnarray}
\normalsize
Using this, we get $\{A_x\}$ in terms of $\{A'_y\}$ as
\begin{equation}
\begin{aligned}
    A'_1 +A'_2 +A'_3 +A'_4 = 8 \ A_1 \\
    A'_1 +A'_2 +A'_3 -A'_4 = 8 \ A_2 \\
    A'_1 +A'_2 -A'_3 +A'_4 = 8 \ A_3 \\
    A'_1 +A'_2 -A'_3 -A'_4 = 8 \ A_4 \\
    A'_1 -A'_2 +A'_3 +A'_4 = 8 \ A_5 \\
    A'_1 -A'_2 +A'_3 -A'_4 = 8 \ A_6 \\
    A'_1 -A'_2 -A'_3 +A'_4 = 8 \ A_7 \\
    A'_1 -A'_2 -A'_3 -A'_4 = 8 \ A_8 
    \end{aligned}
\end{equation}
Now, employing the condition $A_1^2 = A_x^2~ \forall x \in \{2,3,4,5,6,7\}$, we obtain constraints on the anticommutators $\{A'_y ,A'_{y'}\} \equiv A'_{yy'}$. As shown in Eq.~(\ref{aceqn}), these constraints can be expressed in matrix form as
\begin{eqnarray}\label{aceq4}
\mathbf{P}\cdot \mathcal{S}=0 \implies 
    \begin{bmatrix}
        A'_{12} & A'_{13} & A'_{14} & A'_{23} & A'_{24} & A'_{34} 
    \end{bmatrix}
    \begin{bmatrix}
        0 & 0 & 0 & 1 & 1 & 1 & 1 \\
        0 & 1 & 1 & 0 & 0 & 1 & 1 \\
        1 & 0 & 1 & 0 & 1 & 0 & 1 \\
        0 & 1 & 1 & 1 & 1 & 0 & 0 \\
        1 & 0 & 1 & 1 & 0 & 1 & 0\\
        1 & 1 & 0 & 0 & 1 & 1 & 0
    \end{bmatrix}
    =
    \begin{bmatrix}
        0 & 0 & 0 & 0 & 0 & 0 & 0
    \end{bmatrix}
\end{eqnarray}
In this equation, the $6\times 7$ matrix is $\mathcal{S}$ similar to Eq.~(\ref{aceqn}). Now we mimic the row and column operations as illustrated in Fig.~(\ref{co}) to make $\mathcal{S}$ as upper-triangular as shown in the following.

\begin{eqnarray}
    \begin{bmatrix}
        0 & 0 & 0 & 1 & 1 & 1 & 1 \\
        0 & 1 & 1 & 0 & 0 & 1 & 1 \\
        1 & 0 & 1 & 0 & 1 & 0 & 1 \\
        0 & 1 & 1 & 1 & 1 & 0 & 0 \\
        1 & 0 & 1 & 1 & 0 & 1 & 0\\
        1 & 1 & 0 & 0 & 1 & 1 & 0
    \end{bmatrix}
    \longrightarrow\left[
    \begin{array}{ccc|ccc|c}
        1 & 0 & 0 & 1 & 1 & 0 & 1 \\
        0 & 1 & 0 & 1 & 0 & 1 & 1 \\
        0 & 0 & 1 & 0 & 1 & 1 & 1 \\
        \hline
        1 & 1 & 0 & 0 & 1 & 1 & 0 \\
        1 & 0 & 1 & 1 & 0 & 1 & 0\\
        0 & 1 & 1 & 1 & 1 & 0 & 0
    \end{array}\right]
    \longrightarrow\left[
    \begin{array}{ccc|ccc|c}
        1 & 0 & 0 & 1 & 1 & 0 & 1 \\
        0 & 1 & 0 & 1 & 0 & 1 & 1 \\
        0 & 0 & 1 & 0 & 1 & 1 & 1 \\
        \hline
        0 & 0 & 0 & -2 & 0 & 0 & -2 \\
        0 & 0 & 0 & 0 & -2 & 0 & -2\\
        0 & 0 & 0 & 0 & 0 & -2 & -2
    \end{array}\right]
\end{eqnarray}
Thus, being $\mathcal{S}$ is a full rank, the only solution to Eq.~(\ref{aceq4}) is $\{A'_y ,A'_{y'}\} =0$.


\section{Construction of Alice's and Bob's Observables which yield optimal Bell value for two copies of Bell diagonal states for $n=4$} \label{example}

We will use notations $\sigma_i$ and $\sigma_i'$ to denote the observables of Alice and Bob respectively. Using these notations, a two-qubit Bell diagonal state is expressed as
\begin{eqnarray}
    \rho_{D}= \frac{1}{4} \qty(\openone\otimes\openone + \sum_{i=1}^3 \lambda_i \ \qty(\sigma_i \otimes \sigma'_i))
\end{eqnarray}
Two copies of this state $(\rho=  \rho_{D}\otimes\rho_{D})$ is then given by
\begin{equation}
  \rho=  \frac{1}{16} \qty[\Motimes_{4}\openone + \sum_{i=1}^{3} \lambda_i\qty( \sigma_i\otimes \openone)_A \Motimes \qty(\sigma'_i\otimes \openone)_B +  \sum_{j=1}^{3} \lambda_j\qty(\openone\otimes\sigma_j)_A \Motimes \qty(\openone\otimes\sigma'_j)_B + \sum_{m,n=1}^{3} \lambda_m\lambda_n \qty(\sigma_m \otimes\sigma_n)_A \Motimes \qty(\sigma'_m \otimes\sigma'_n)_B]
\end{equation}
The correlation matrix of this state is expressed as
\begin{eqnarray}
T_{\rho}=
    \begin{bmatrix}
        \lambda_1 & 0 & 0 & 0 & 0 & 0 & 0 & 0 & 0 & 0 & 0 & 0 & 0 & 0 & 0 \\
        0 & \lambda_2 & 0 & 0 & 0 & 0 & 0 & 0 & 0 & 0 & 0 & 0 & 0 & 0 & 0 \\
        0 & 0 & \lambda_3 & 0 & 0 & 0 & 0 & 0 & 0 & 0 & 0 & 0 & 0 & 0 & 0 \\
        0 & 0 & 0 & \lambda_1 & 0 & 0 & 0 & 0 & 0 & 0 & 0 & 0 & 0 & 0 & 0 \\
        0 & 0 & 0 & 0 & \lambda_2 & 0 & 0 & 0 & 0 & 0 & 0 & 0 & 0 & 0 & 0 \\
        0 & 0 & 0 & 0 & 0 & \lambda_3 & 0 & 0 & 0 & 0 & 0 & 0 & 0 & 0 & 0 \\
        0 & 0 & 0 & 0 & 0 & 0 & \lambda_1^2 & 0 & 0 & 0 & 0 & 0 & 0 & 0 & 0 \\
        0 & 0 & 0 & 0 & 0 & 0 & 0 & \lambda_1 \lambda_2 & 0 & 0 & 0 & 0 & 0 & 0 & 0 \\
        0 & 0 & 0 & 0 & 0 & 0 & 0 & 0 & \lambda_1 \lambda_3 & 0 & 0 & 0 & 0 & 0 & 0 \\
        0 & 0 & 0 & 0 & 0 & 0 & 0 & 0 & 0 & \lambda_1 \lambda_2 & 0 & 0 & 0 & 0 & 0 \\
        0 & 0 & 0 & 0 & 0 & 0 & 0 & 0 & 0 & 0 & \lambda_2^2 & 0 & 0 & 0 & 0 \\
        0 & 0 & 0 & 0 & 0 & 0 & 0 & 0 & 0 & 0 & 0 & \lambda_2 \lambda_3 & 0 & 0 & 0 \\
        0 & 0 & 0 & 0 & 0 & 0 & 0 & 0 & 0 & 0 & 0 & 0 & \lambda_1 \lambda_3 & 0 & 0 \\
        0 & 0 & 0 & 0 & 0 & 0 & 0 & 0 & 0 & 0 & 0 & 0 & 0 & \lambda_2 \lambda_3 & 0 \\
        0 & 0 & 0 & 0 & 0 & 0 & 0 & 0 & 0 & 0 & 0 & 0 & 0 & 0 & \lambda_3^2
    \end{bmatrix}
\end{eqnarray} 
The six sets of anticommuting observables in the operator space acting on $\mathscr{H}_2 \otimes \mathscr{H}_2$ is given by
\begin{equation} \label{acset}
\begin{aligned}
\mathcal{S}_1 &=\lbrace\sigma_{x}\otimes\sigma_{x},\sigma_{x}\otimes\sigma_{y},\sigma_{x}\otimes\sigma_{z},\sigma_{y}\otimes\openone, \sigma_{z}\otimes\openone\rbrace \ ; \ \ \mathcal{S}_2 = \lbrace\sigma_{y}\otimes\sigma_{x},\sigma_{y}\otimes\sigma_{y},\sigma_{y}\otimes\sigma_{z},\sigma_{x}\otimes\openone, \sigma_{z}\otimes\openone\rbrace \\ 
\mathcal{S}_3 &=\lbrace\sigma_{z}\otimes\sigma_{x},\sigma_{z}\otimes\sigma_{y},\sigma_{z}\otimes\sigma_{z},\sigma_{y}\otimes\openone, \sigma_{x}\otimes\openone\rbrace \ ; \ \ \mathcal{S}_4 = \lbrace\sigma_{x}\otimes\sigma_{x},\sigma_{y}\otimes\sigma_{x},\sigma_{z}\otimes\sigma_{x},\openone\otimes\sigma_{y}, \openone\otimes\sigma_{z}\rbrace\\
\mathcal{S}_5 &= \lbrace\sigma_{x}\otimes\sigma_{y},\sigma_{y}\otimes\sigma_{y},\sigma_{z}\otimes\sigma_{y},\openone\otimes\sigma_{x}, \openone\otimes\sigma_{z}\rbrace  \ ; \ \ \mathcal{S}_6 = \lbrace\sigma_{x}\otimes\sigma_{z},\sigma_{y}\otimes\sigma_{z},\sigma_{z}\otimes\sigma_{z},\openone\otimes\sigma_{x}, \openone\otimes\sigma_{y}\rbrace
\end{aligned}
\end{equation}
Without loss of generality, we assume $\lambda_1 \leq \lambda_2 \leq \lambda_3$. The maximisation in LHS of Eq.~(\ref{belldiag}) in the main text is
\begin{equation}
     \max_{\mathcal{S}^{(n)}_l} \sqrt{\sum_{i=1}^n \mu_i} = \sqrt{\lambda_3^2 |\lambda|^2 + \lambda_2^2} \ ; \ \ \ \text{where $|\lambda|^2= \lambda_1^2 +\lambda_2^2 +\lambda_3^2$}
\end{equation}
Recalling Proposition \ref{prop}, the scaled observables of Alice $\{\mathcal{A}_y \}$ will belong to one of these sets. We see that the maximum value over all sets is attained for the following choice of observables from $\mathcal{S}_3$ in Eq.~(\ref{acset})
\begin{equation}
\begin{aligned}
    \mathcal{A}_1 &= \sigma_{z}\otimes \sigma_{x} \ ; \ \ \mathcal{A}_2 = \sigma_{z}\otimes\sigma_{y} \ ; \ \ \mathcal{A}_3 = \sigma_{z}\otimes\sigma_{z} \ ; \ \ \mathcal{A}_4 = \sigma_{y}\otimes\openone \ ;  \\
    B_1 &= \sigma_{z}'\otimes\sigma_{x}' \ ; \ \ B_2 = \sigma_{z}'\otimes\sigma_{y}' \ ; \ \ B_3 =\sigma_{z}' \otimes\sigma_{z}' \ ; \ \ B_4 = \sigma_{y}'\otimes\openone \ ; \\
    \omega_1 &= \frac{8 \lambda_1 \lambda_3}{\sqrt{\lambda_2^2 + \lambda_3^2 |\lambda|^2}} \ ; \ \ \omega_2 = \frac{8 \lambda_2 \lambda_3}{\sqrt{\lambda_2^2 + \lambda_3^2 |\lambda|^2}} \ ; \ \ \omega_3 = \frac{8\lambda_3^2}{\sqrt{\lambda_2^2 + \lambda_3^2 |\lambda|^2}} \ ; \ \ \omega_4 = \frac{8 \lambda_2 }{\sqrt{\lambda_2^2 + \lambda_3^2 |\lambda|^2}}
    \end{aligned}
\end{equation}
From which we can get $\{A_x\}$ in terms of $\{A'_y =\omega_y \mathcal{A}_y \}$ as
\begin{equation}
\begin{aligned}
    A_1 &= \frac{1}{\sqrt{\lambda_2^2 + \lambda_3^2 |\lambda|^2}}\qty(\lambda_3 \sigma_{z}\otimes(\lambda_1 \sigma_x +\lambda_2 \sigma_y +\lambda_3 \sigma_z) + \lambda_2 \sigma_{y}\otimes\openone) \\
    A_2 &= \frac{1}{\sqrt{\lambda_2^2 + \lambda_3^2 |\lambda|^2}}\qty(\lambda_3 \sigma_z \otimes(\lambda_1 \sigma_x +\lambda_2 \sigma_y +\lambda_3 \sigma_z) - \lambda_2 \sigma_y \otimes \openone) \\
    A_3 &= \frac{1}{\sqrt{\lambda_2^2 + \lambda_3^2 |\lambda|^2}}\qty(\lambda_3 \sigma_z \otimes(\lambda_1 \sigma_x +\lambda_2 \sigma_y -\lambda_3 \sigma_z) + \lambda_2 \sigma_y \otimes \openone) \\
    A_4 &= \frac{1}{\sqrt{\lambda_2^2 + \lambda_3^2 |\lambda|^2}}\qty(\lambda_3 \sigma_z \otimes(\lambda_1 \sigma_x +\lambda_2 \sigma_y -\lambda_3 \sigma_z) - \lambda_2 \sigma_y \otimes \openone) \\
    A_5 &= \frac{1}{\sqrt{\lambda_2^2 + \lambda_3^2 |\lambda|^2}}\qty(\lambda_3 \sigma_z \otimes(\lambda_1 \sigma_x -\lambda_2 \sigma_y +\lambda_3 \sigma_z) + \lambda_2 \sigma_y \otimes \openone) \\
    A_6 &= \frac{1}{\sqrt{\lambda_2^2 + \lambda_3^2 |\lambda|^2}}\qty(\lambda_3 \sigma_z \otimes(\lambda_1 \sigma_x -\lambda_2 \sigma_y +\lambda_3 \sigma_z) - \lambda_2 \sigma_y \otimes \openone) \\
    A_7 &=\frac{1}{\sqrt{\lambda_2^2 + \lambda_3^2 |\lambda|^2}}\qty(\lambda_3 \sigma_z \otimes(\lambda_1 \sigma_x -\lambda_2 \sigma_y -\lambda_3 \sigma_z) + \lambda_2 \sigma_y \otimes \openone) \\
    A_8 &=\frac{1}{\sqrt{\lambda_2^2 + \lambda_3^2 |\lambda|^2}}\qty(\lambda_3 \sigma_z\otimes(\lambda_1 \sigma_x -\lambda_2 \sigma_y -\lambda_3 \sigma_z) - \lambda_2 \sigma_y \otimes \openone) 
    \end{aligned}
\end{equation}


\end{document}